\documentclass{fundam}
\usepackage[english]{babel}
\usepackage[utf8]{inputenc}
\usepackage{amssymb}
\usepackage[T1]{fontenc}
\usepackage[single]{accents}
\usepackage{graphicx}
\usepackage{tikz}
\usetikzlibrary{calc, positioning, arrows.meta, fit, scopes}
\usepackage{xspace}

\usepackage{mymacros}

\begin{document}

\setcounter{page}{31}
\publyear{22}
\papernumber{2129}
\volume{187}
\issue{1}

  \finalVersionForARXIV

\title{Number Conservation via Particle Flow in \\ One-dimensional Cellular Automata}

\author{Markus Redeker\thanks{Address for correspondence:   markus2.redeker@mail.de \newline \newline
          \vspace*{-6mm}{\scriptsize{Received July 2021; \ accepted August  2022.}}}
                   \\
  Hamburg, Germany \\
  markus2.redeker@mail.de}

\runninghead{M. Redeker}{Number Conservation via Particle Flow in One-dimensional Cellular Automata}

\maketitle

\begin{abstract}
  A number-conserving cellular automaton is a simplified model for a
  system of interacting particles. This paper contains two related
  constructions by which one can find all one-dimensional
  number-conserving cellular automata with one kind of particle.

  The output of both methods is a ``flow function'', which describes
  the movement of the particles. In the first method, one puts
  increasingly stronger restrictions on the particle flow until a
  single flow function is specified. There are no dead ends, every
  choice of restriction steps ends with a flow.

  The second method uses the fact that the flow functions can be
  ordered and then form a lattice. This method consists of a recipe
  for the slowest flow that enforces a given minimal particle speed in
  one given neighbourhood. All other flow functions are then maxima of
  sets of these flows.

  Other questions, like that about the nature of non-deterministic
  number-conserving rules, are treated briefly at the end.
\end{abstract}

\section{Introduction}

Cellular automata are microscopic worlds: extremely simple spaces in
which time passes and events occur. Sometimes they are used to
simulate aspects of this world. It is therefore an interesting
question to ask which of these micro-worlds can be interpreted as
containing particles that can move, collide or stick to each other,
but are neither destructed nor created.

In this paper I provide a new answer to this question, valid for
one-dimensional automata with only one kind of particles. The new idea
is that such cellular automata are now \emph{constructed} instead of
testing whether a given system satisfies the requirements.

The cellular automata are constructed in terms of \emph{flow
  functions}. These are functions that describe how many particles
cross the boundary between two cells, depending on the neighbourhood
of this boundary. By specifying the flow functions, one can find all
number-conserving automata.

Flow functions did already occur in the literature. They were
described by Hattori and Takesue \cite[Th.~2.2]{HattoriTakesue1991}
and by Pivato \cite[Prop.~12]{Pivato2002} but not used for the
construction of transition rules for cellular automata. Imai and
Alhazov \cite{ImaiAlhazov2010} have defined flow functions for
cellular automata of radius $\frac12$ and~1, and their Proposition~3
contains a necessary and sufficient condition for number-conservation
in radius~$\frac12$ cellular automata -- it is the case of $\ell = 1$
for Theorem~\ref{thm:flow-construction} of this paper, so to speak.

This paper contains two related constructions for flow functions (and
therefore number-conserving cellular automata). The first one works by
stepwise refinement. It starts with very weak requirements on the flow
function, which then are successively strengthened until a single
function remains that satisfies all requirements. As a side effect,
the construction gives an overview over the set of all
number-conserving cellular automata; this could be used for e.\,g.\
the classification or enumeration of the automata.

To write about particle flows, one needs an appropriate notation. It
should represent all flows and only them, and ideally it should relate
in a meaningful way to the represented flows. So far, the flows were
represented by functions from cellular neighbourhoods to the integers,
which violated the second requirement because flow functions cannot
easily be distingushed from non-flow functions.

To correct this, a second construction is introduced, based on the
first one. It uses the fact that the flow functions form a lattice and
the minimal elements of this lattice have a reasonably simple form: An
arbitrary flow function can therefore always written as the maximum of
some of these minimal flows.

\medskip\smallskip\noindent \textbf{Background} \
The following paragraphs contain references to papers that are
somewhat related to this work. The list is by no means complete. For a
little bit of history of number conservation see Bhattacharjee \etal\
\cite[Sec.~4.6]{BhattacharjeeNaskarEtAl2016} or the introduction of
the related paper by Wolnik \etal\ \cite{WolnikEtAl2019}.

The most important predecessor of the current work is the paper of
Hattori and Takesue \cite{HattoriTakesue1991} about additive
invariants of cellular automata. The number of particles is one of
them. Hattori and Takesue found a simple way to test whether a
one-dimensional cellular automaton is number-conserving. This method
works only \emph{a posteriori}, but it provides the base for most
later papers about number conservation.

Boccara and Fuk\'s \cite{Boccara1998} describe the behaviour of
number-conserving cellular automata with the help of ``motion
representation'' diagrams. These are diagrams that show the motion of
individual particles in finite regions of the automaton. The authors
find a set of equations for the transition rule of a number-conserving
automaton and solve them. They also show that number-conserving rules
can be identified by verifying that the number of particles is
conserved on all circular cell configurations of a specified size.

Fuk\'s \cite{Fuks2000} shows that in one-dimensional number-conserving
automata, one can assign permanent identities to the ``particles'' and
describe the evolution of the cellular automaton in terms of particle
movements.

Pivato \cite{Pivato2002} derives several characterisations of
conservation laws in the more general context of cellular automata on
arbitrary groups.

Durand, Formenti and Róka \cite{DurandFormentiEtAl2003} collect
different definitions of number conservation and show that they are
equivalent. They also derive conditions for number conservation in
cellular automata of any dimension.

\paragraph{Overview} The rest of this article, after the introduction
and a section with definitions, consists of the following major parts.

In Section~\ref{sec:flow-conditions}, the \emph{flow function} for a
number-conserving cellular automaton is defined and the \emph{flow
  conditions} are derived: conditions on the flow functions that are
satisfied if and only if the functions belong to a number-conserving
cellular automaton.

In Section~\ref{sec:solve-flow-conditions}, all solutions to the flow
conditions are found. Examples and a diagram notation for the flow
function follow in the next section.

In Section~\ref{sec:minimal-flows}, a lattice structure for the set of
flow functions is found. This leads to the notion of a set of
``minimal'' flow functions from which all other flows can be built
with the help of lattice operations. A recipe for minimal flow
functions is found and examples for minimal flows are shown.

Two sections, one about related topics and one with open questions,
follow at the end.

\section{Definitions}

A \emph{one-dimensional cellular automaton} is a discrete dynamical
system; its states are configurations of simpler objects, the
\emph{cells}. The cells are arranged in an infinite line -- they are
indexed by $\Z$ -- and the state of each cell is an element of a
finite set $\Sigma$.

\subsection{Cells and states}

In a number-conserving cellular automaton, we picture each cell as a
container for a certain number of particles; the number of particles
it contains is part of its state. Since the number of states is
finite, there is a maximal number $C$ of particles that a cell may
contain, the \emph{capacity} of the cellular automaton. We therefore
have for each cell state $\alpha \in \Sigma$ a number
$\# \alpha \in \{ 0, \dots, C \}$, the \emph{particle content} of
$\alpha$. The expression
\begin{equation}
  \label{eq:complement}
  \#^c \alpha = C - \#\alpha
\end{equation}
stands for the \emph{complement} of the particle content: It is the
maximal number of particles one can still put into a cell of state
$\alpha$ without exceeding its capacity.

\medskip
To make things simpler, the mapping
$\# \colon \Sigma \to \{ 0, \dots, C \}$ is required to be surjective.
On the other hand, it is not necessary that $\#$ is injective; the
constructions in this paper will work even if this is not the case.

Sometimes we will take the numbers $\{ 0, \dots, C \}$ directly as
cell states. Such a state set is called a \emph{minimal set} of
states. Even then, we will keep the distinction between $\alpha$ and
$\# \alpha$, to make it clearer whether we are speaking of cells or of
the number of particles in them.

\subsection{Configurations}

We will need to speak about finite and double infinite sequences of
cell states. Finite sequences are also called \emph{strings}, while
the infinite sequences are \emph{configurations}.

A \emph{string} is an element of
$\Sigma^* = \bigcup_{k\geq 0} \Sigma^k$. If $a \in \Sigma^k$ is a
string, then $k$ is its \emph{length}, and we write $\abs{a}$ for the
length of $a$. For the string of length 0 we write $\epsilon$.

\medskip
Often we write a string $a \in \Sigma^k$ as a product of cell states,
in the form
\begin{equation}
  \label{eq:finite-string}
  a = a_0 \dots a_{k-1}\,.
\end{equation}
The number of particles in $a$ and its complement are then
\begin{equation}
  \label{eq:finite-string-particles}
  \# a = \sum_{i=0}^{k-1} \# a_i
  \qquad\text{and}\qquad
  \#^c a = \abs{a} C - \#a\,.
\end{equation}
We will also need substrings of $a$. They are specified by start point
and length, in the form
\begin{equation}
  \label{eq:substring}
  a_{m:n} = a_m \dots a_{m + n - 1}\,.
\end{equation}

\medskip
A \emph{configuration} of a cellular automaton is a doubly infinite
sequence of cell states, in the form
\begin{equation}
  \label{eq:configuration}
  a = \dots a_{-3} a_{-2} a_{-1} a_0 a_1 a_2 a_3 \dots \,.
\end{equation}
The conventions for substrings and particle content are the same for
configurations as for strings, except that $\# a$ is only defined when
the number of the $a_i$ that have a non-zero content is finite. Such a
configuration is said to have a \emph{finite particle content}.

\subsection{Evolution}

The configuration of a cellular automaton changes over time. This
behaviour -- the \emph{evolution} of the automaton -- is specified by
the \emph{global transition rule} $\hat{\phi}$, which maps the
configuration at time $t$ to the configuration at time $t + 1$.

As a function, $\hat{\phi}$ is determined by two integers $r_1$ and
$r_2$ and a \emph{local transition rule}
$\phi \colon \Sigma^{r_1 + r_2 + 1} \to \Sigma$, subject only to the
condition that $r_1 + r_2 \geq 0$. The numbers $r_1$ and~$r_2$ are the
\emph{left} and \emph{right radius} of $\phi$. (Almost always we will
also require that $r_1$ and $r_2$ are non-negative numbers. The only
exceptions are the shift rules in the example below.)

\medskip
Now we can express the value $\hat{\phi}(a)$ for an arbitrary
configuration $a$ by the condition that
\begin{equation}
  \label{eq:local-evolution}
  \hat{\phi}(a)_x = \phi(a_{x - r_1}, \dots, a_{x + r_2})
  \quad\text{for all $x \in \Z$.}
\end{equation}
If we want to use the colon form for substrings that is defined
in~(\ref{eq:finite-string-particles}), we can write this rule also in
the form $\hat{\phi}(a)_x = \phi(a_{x - r_1: r_1 + r_2 + 1})$. This
form, which emphasises the length of the local neighbourhood but is
less symmetric, will soon be used.

\medskip
(Often one only considers the symmetric case with $r_1 = r_2 = r$.
Then $r$ is called the \emph{radius} of the transition rule.)

\subsection{Rules with equivalent dynamics}
\label{sec:equivalent-dynamics}

When we create a new transition rule $\phi_v$ by replacing the radii
$r_1$ and $r_2$ of $\phi$ with $r'_1 = r_1 + v$ and $r'_2 = r_2 - v$,
the behaviour of the new transition rule $\hat{\phi}_v$ does not
differ significantly from that of $\hat{\phi}$, except for a
horizontal shift that occurs at each time step.

\medskip
It has then the dynamics
$\hat{\phi}_v(a)_x = \phi(a_{x - r_1 - v: r_1 + r_2 + 1})$, or
equivalently,
\begin{equation}
  \label{eq:phi-variation}
  \hat{\phi}_v(a)_{x + v} = \phi(a_{x - r_1: r_1 + r_2 + 1})
  \qquad\text{for all $x \in \Z$.}
\end{equation}
At every time step, the rule $\hat{\phi}_v$ moves therefore the new
cell states $v$ positions more to the right than $\hat{\phi}$. We can
say that $\hat\phi_v$ is $\hat\phi$ as seen by an observer who moves
with speed $v$ to the left.

All rules $\phi_v$ can be treated in essentially the same way,
independent of the value of $v$. We therefore introduce a new
parameter $\ell = r_1 + r_2$ to characterise $\phi$. The number $\ell$
is, anticipatingly, called the \emph{flow length} of
$\phi$.\footnote{At this point, a parameter with a value of
  $r_1 + r_2 + 1$ might look more natural, but we will see that
  $r_1 + r_2$ occurs more often in our calculations.}

\paragraph{Shift rules} The simplest example for rules with the same
dynamics are the \emph{shift rules} $\hat{\sigma}^k$, which exist for
all $k \in \Z$. The rule $\hat{\sigma}^k$ moves in each time step the
states of all cells by $k$ positions to the right.

They have $r_1 = k$, $r_2 = -k$ (therefore $\ell = 0$) and a local
transition rule of the form $\sigma^k \colon \Sigma \to \Sigma$. Since
all shift rules have the same dynamics, all $\sigma^k$ must be the
same function: It is the identity. All $\hat{\sigma}^k$ are trivially
number-conserving, independent of the particle contents of the cell
states.

This is not the only representation of the shift rules as cellular
automata. We will soon see another representation for the rules
$\hat\sigma^k$ with $k \geq 0$.

\section{The flow conditions}
\label{sec:flow-conditions}

A formal definition for number conservation is still missing. We use
this one:
\begin{quote}
  A transition rule $\hat{\phi}$ is \emph{number-conserving} if
  $\# \hat{\phi}(a) = \# a$ for every configuration $a$ with finite
  particle content.
\end{quote}
From now on we will therefore assume that every configuration has a
finite particle content.

If we divide such a configuration into two parts, we can define the
\emph{particle flow} between them. It is the number of particles that
move from the left to the right side under the action of $\hat{\phi}$.

\medskip
For the definition, consider the following setup:
\begin{equation}
  \label{eq:particle-flow-0}
  \dots u_{-3} u_{-2} u_{-1}
  \boundary
  v_0 v_1 v_2 \dots\,.
\end{equation}
It shows the cells of a configuration, split by the vertical bar into
a left part, $u$, and a right part, $v$. Its image under $\hat \phi$
shall be
\begin{equation}
  \label{eq:particle-flow-1}
  \dots u'_{-3} u'_{-2} u'_{-1}
  \boundary
  v'_0 v'_1 v'_2 \dots\,.
\end{equation}
Let now $L = \# u' - \# u$ be the change in the particle content in
the left part and $R = \# v' - \# v$ be the change in the right part.
The particle content in the configuration is preserved, therefore is
$R = -L$. The number $R$ is then the particle flow; if it is positive,
particles move to the right.

\medskip
We apply this concept to the following setup, which is the same
as~\eqref{eq:particle-flow-0}, but with the cells named differently.
\begin{equation}
  \label{eq:bipartite-setup}
  \dots x_{-2} x_{-1} \, u_0 \dots u_{r_1 - 1}
  \boundary
  u_{r_1} \dots u_{\ell - 1} \, y_\ell y_{\ell + 1} \dots\,.
\end{equation}
The two parts are separated by a vertical bar as before, but but there
are now also three regions, $x$, $u$, and $y$. Region $u$ consists of
$r_1$ cells to the right and $r_2$ cells to the left of the vertical
bar. The numbers $L$ and $R$ are defined as before.

\medskip
The interaction in a cellular automaton is local, therefore $L$ can
only depend on $x$ and $u$, and $R$ only on $u$ and $y$. But
$R = - L$, therefore they both can only depend on~$u$. So we can find
a function
\begin{equation}
  \label{eq:flow-function}
  f \colon \Sigma^\ell \to \Z
\end{equation}
with $f(u) = R = -L$. This is the function which will allow us to
describe the essential properties of a number-conserving cellular
automaton. We call $f$ the \emph{flow function} of~$\phi$.

\medskip
Next we try to reconstruct $\phi$ from $f$. To do this, we consider
the following setup with $\ell + 1$ named cells and two boundaries,
\begin{equation}
  \label{eq:tripartite-setup}
  \dots w_0 \dots w_{r_1 - 1}
  \boundary w_{r_1} \boundary
  w_{r_1 + 1} \dots w_\ell \dots\,.
\end{equation}
We are interested in the state of the cell at the centre in the next
time step. Initially, it contains $\# w_{r_1}$ particles. One step
later, there are $\# \phi(w)$ particles at this place. On the other
hand, during this transition, $f(w_{0:\ell})$ particles must have
entered the cell region through the left boundary, and $f(w_{1:\ell})$
of them must have left it to the right. The number of particles at the
centre at the next time step must therefore be
\begin{equation}
  \label{eq:particles-next-step}
  \# \phi(w) = f(w_{0:\ell}) + \# w_{r_1} - f(w_{1:\ell})\,.
\end{equation}
With this equation, applied to all neighbourhoods
$w \in \Sigma^{\ell + 1}$, we can partially reconstruct $\phi$ from
$f$. If the state set $\Sigma$ is minimal, the transition function can
even be reconstructed uniquely from the values of $\# \phi(w)$.
Otherwise, there are several different transition functions for the
same flow function. (How they are related would be the subject of
another paper.)

\medskip
Since $\phi$ can be derived from $f$, it will be enough to consider
$f$ alone. We therefore need to find all functions
$f \colon \Sigma^\ell \to \Z$ for which there is a valid transition
rule $\phi$. These are exactly those functions $f$ for which the right
side of~\eqref{eq:particles-next-step} is neither too small nor too
large. We have therefore

\begin{theorem}[Flow conditions]
  \label{thm:flow-conditions}
  A number-conserving cellular automaton for a given function
  $f \colon \Sigma^\ell \to \Z$ exists if and only if
  \begin{equation}
    \label{eq:flow-conditions}
    0 \leq f(w_{0:\ell}) + \# w_{r_1} - f(w_{1:\ell}) \leq C
    \qquad\text{for all $w \in \Sigma^{\ell + 1}$.}
  \end{equation}
\end{theorem}
The inequalities~\eqref{eq:flow-conditions} are called the \emph{flow
  conditions}.

\section{Solving the flow conditions}
\label{sec:solve-flow-conditions}

\begin{flushright}\itshape
  Great flows have little flows next to them to guide 'em, \\
  And little flows have lesser flows, but not ad infinitum.
\end{flushright}
\noindent
We will now restrict out work to rules with $r_1 = \ell$ and
$r_2 = 0$, so that particles only move to the right. This will make
induction on the neighbourhood size much simpler. It does not reduce
the generality of the conclusions because, as we have seen in
Section~\ref{sec:equivalent-dynamics}, every rule is equivalent to
such a rule. And at the end, in Section~\ref{sec:two-sided}, we will
return to the general case.

\medskip
The flow conditions now have the form
\begin{equation}
  \label{eq:forward-flow-conditions}
  0 \leq f(w_{0:\ell}) + \# w_\ell - f(w_{1:\ell}) \leq C
  \qquad\text{for all $w \in \Sigma^{\ell + 1}$.}
\end{equation}
To find solutions for them, we define some functions related to $f$,
the \emph{half-flows}. They come in two variants, as \emph{bound} and
\emph{free} half-flows, where the free half-flows are a generalisation
of the bound half-flows.

\medskip
The bound half-flows have their name from the fact that they are
constructed from a specific flow function $f$. This makes their
definition easier to understand and therefore they are introduced
first. Properties of the bound half-flows will then lead to a
definition for free half-flows, which do not refer to a flow function.
Instead, each system of free half-flows is used to \emph{construct} a
flow function.

\paragraph{Bound half-flows} We will now introduce the bound
half-flows. There are two families of them, the \emph{lower} and the
\emph{upper half-flows}, given by the equations,
\begin{subequations}
  \label{eq:half-flows}
  \begin{align}
    \label{eq:lower-half-flow}
    \fminb_k(v) &= \min \set{f(u v)\colon u \in \Sigma^{\ell - k}}, \\
    \label{eq:upper-half-flow}
    \fmaxb_k(v) &= \max \set{f(u v)\colon u \in \Sigma^{\ell - k}},
  \end{align}
\end{subequations}
with $0 \leq k \leq \ell$ and $v \in \Sigma^k$. (Note that $\fmaxb_0$
and $\fminb_0$ are in fact constants.) Much more useful are however
the inductive forms of these definitions,
\begin{subequations}
  \label{eq:half-flows-rec}
  \begin{align}
    \label{eq:lower-half-flow-rec}
    \fminb_\ell(w) &= f(w), &
    \fminb_k(v) &= \min \set{\fminb_{k+1}(\alpha v)\colon \alpha \in \Sigma}, \\
    \label{eq:upper-half-flow-rec}
    \fmaxb_\ell(w) &= f(w), &
    \fmaxb_k(v) &= \max \set{\fmaxb_{k+1}(\alpha v)\colon \alpha \in \Sigma},
  \end{align}
\end{subequations}
with $w \in \Sigma^\ell$, $\alpha \in \Sigma$ and $v \in \Sigma^k$ for
$1 \leq k \leq \ell$.

\paragraph{Example: Flow functions for shift rules} The shift rules
$\hat \sigma^\ell$ for $\ell \geq 0$ illustrate these definitions.
Each shift rule has a flow function
$f \colon \Sigma^\ell \to \{ 0, \dots, \ell C \}$ with $f(v) = \# v$
for all $v$. Its half-flows are $\fminb_k(w) = \# w$ and
$\fmaxb_k(w) = (\ell - k)C + \# w$, for $w \in \Sigma^k$. This shows
that the particle flow can be quite large, greater than the capacity
of a single cell.

\subsection{Properties of the bound half-flows}

The most important consequence of~\eqref{eq:half-flows-rec} is the
following lemma, which shows how the half-flows are related to each
other.

\medskip
It especially shows that for $1 \leq k \leq \ell$, the flow
conditions~\eqref{eq:forward-flow-conditions} split into pairs of
inequalities,
\begin{subequations}
  \label{eq:capacities}
  \begin{align}
    \label{eq:lower-capacities}
    0 \leq {} & \fminb_k(w_{0:k}) + \# w_k - \fminb_k(w_{1:k}), \\
    \label{eq:upper-capacities}
              & \fmaxb_k(w_{0:k}) + \# w_k - \fmaxb_k(w_{1:k}) \leq C.
  \end{align}
\end{subequations}
In the lemma, these inequalities follow from
in~\eqref{eq:raising-recursion}: We can get them by removing the
central terms of~\eqref{eq:lower-raising-recursion}
and~\eqref{eq:upper-raising-recursion} and rearranging the remaining
inequalities.

\begin{lemma}[Interaction of half-flows]
  \label{thm:half-flows-recursion}
  Let $0 \leq k < \ell$ and $w \in \Sigma^{k+1}$. If $f$ is the flow
  function of a number-conserving cellular automaton, then
  \begin{subequations}
    \label{eq:raising-recursion}
    \begin{alignat}{2}
      \label{eq:lower-raising-recursion}
      \fminb_k(w_{1:k}) \leq \fminb_{k+1}(w)
      &\leq \fminb_k(w_{0:k}) + \# w_k, \\
      \label{eq:upper-raising-recursion}
      \fmaxb_k(w_{0:k}) - \#^c w_k \leq \fmaxb_{k+1}(w)
      &\leq \fmaxb_k(w_{1:k})\,.
    \end{alignat}
  \end{subequations}
\end{lemma}

\begin{proof}
  The two pairs of inequalities can be proved independently of each
  other, so we begin with~\eqref{eq:lower-raising-recursion}.

\medskip
  Its proof is a finite induction from $k = \ell$ down to $k = 1$. The
  induction step consists of showing that that from
  \begin{alignat}{2}
    \label{eq:induction-step-1}
    \fminb_k(w_{1:k}) & \leq \fminb_k(w_{0:k}) + \# w_k
    &&\qquad\text{for all $w \in \Sigma^{k+1}$} \\
    \intertext{always follows}
    \label{eq:induction-step-2}
    \fminb_{k-1}(w'_{1:k-1})
    \leq \fminb_k(w')
    &\leq \fminb_{k-1}(w'_{0:k-1}) + \# w'_{k-1}
    &&\qquad\text{for all $w' \in \Sigma^k$.}
  \end{alignat}
  The induction can begin because for $k = \ell$,
  inequality~\eqref{eq:induction-step-1} is equivalent to the left
  side of the flow condition~\eqref{eq:forward-flow-conditions}.

\medskip
  To prove the induction step, note that the left side
  of~\eqref{eq:induction-step-1} is independent of the choice of
  $w_0$. At the right side of~\eqref{eq:induction-step-1}, we can
  therefore replace the term $\fminb_k(w_{0:k})$ with the smallest
  possible value it can take when we vary $w_0$ and keep $w_{1:k-1}$
  fixed. The result is $\fminb_{k-1}(w_{1:k-1})$, and we get
  $\fminb_k(w_{1:k}) \leq \fminb_{k-1}(w_{1:k-1}) + \# w_k$. This is
  already the right inequality of~\eqref{eq:induction-step-2}; we only
  need to write $w_{1:k}$ as $w'$. The left inequality,
  $\fminb_{k-1}(w'_{1:k-1}) \leq \fminb_k(w')$, follows directly from
  the inductive definition~\eqref{eq:lower-half-flow-rec} of~$\fminb_k$.

  The proof of~\eqref{eq:upper-raising-recursion} is similar. We only
  need to replace $\fminb$ with $\fmaxb$, $\#\beta$ with $-\#^c\beta$
  and revert the order of the terms in the inequalities.
\end{proof}

The next lemma is about the behaviour of $\fminb_k(v)$ and $\fmaxb_k(v)$
for a single value of $v$.

\begin{lemma}[Bounds on the half-flows]
  \label{thm:half-flow-bounds}
  Let $f$ be the flow function of a number-conserving cellular
  automaton. Then
  \begin{subequations}
    \label{eq:flow-bounds}
    \begin{align}
      \label{eq:fmin-bounds}
      0 \leq \fminb_k(v) &\leq \#v, \\
      \label{eq:fmax-bounds}
      0 \leq \fmaxb_k(v) &\leq (\ell - k)C + \#v
    \end{align}
  \end{subequations}
  and
  \begin{equation}
    \label{eq:fmin-fmax-bounds}
    \fminb_k(v) \leq \fmaxb_k(v) \leq \fminb_k(v) + (\ell - k)C
  \end{equation}
  are valid for $0 \leq k \leq \ell$ and $v \in \Sigma^k$.
\end{lemma}

\begin{proof}
  The first two inequalities follow directly from the definitions of
  $\fminb_k(v)$ and $\fmaxb_k(v)$ in~\eqref{eq:half-flows}. For both
  half-flows, we have to consider all flows $f(u v)$ with
  $u \in \Sigma^{\ell-k}$, i.\,e.\ the following setup:
  \begin{equation}
    \label{eq:half-flow-setup}
    \dots u_1 \dots u_{\ell-k} v_1 \dots v_k \boundary \dots
  \end{equation}
  Only the particles in $u$ and $v$ may cross the boundary, and they
  may only move to the right. All possible values for $f(u v)$ are
  therefore non-negative, and the same is true for $\fminb_k(v)$ and
  $\fmaxb_k(v)$. This proves the two lower bounds
  in~\eqref{eq:flow-bounds}. For the upper bound on $\fminb_k(v)$, we
  note that the set of all $f(u v)$ includes the case with $\#u = 0$.
  Then at most $\#v$ particles may cross to the right, which means
  that also $\fminb_k(v) \leq \#v$. On the other hand it is also
  possible that all cells in $u$ contain $C$ particles and that
  $\# u = (l - k) C$. Then $(\ell - k)C + \#v$ particles may cross the
  boundary, which explains the upper bound for $\fmaxb_k(v)$.

\medskip
  The left side of~\eqref{eq:fmin-fmax-bounds} is clear; the right
  side can be proved with another setup:
  \begin{equation}
    \label{eq:upper-bound-setup}
    \dots v_1 \dots v_k \boundary
    w_1 \dots w_{\ell-k} \boundary \dots
  \end{equation}
  Here we assume that we know $f(v w)$ and try to find upper and lower
  bounds for the number of particles that cross the left boundary. The
  highest possible number of particles is $f(v w) + \#^c w$, because
  then the $w$ region will be completely filled up in the next time
  step. The smallest particle flow is $f(v w) - \# w$, because then
  the $w$ region will be empty. This means that
  $\fmaxb_k(v) \leq f(v w) + \#^c w$ and
  $\fminb_k(v) \geq f(v w) - \# w$. When we subtract these
  inequalities, we get
  $\fmaxb_k(v) - \fminb_k(v) \leq \#^c w + \# w = (\ell - k) C$, from
  which the right side of~\eqref{eq:fmin-fmax-bounds} follows.
\end{proof}

\subsection{The construction of all flows}
\label{sec:all-flows}

Now we can define the free half-flows. They provide us with a
practical method to find the flow functions among all the functions
$f \colon \Sigma^\ell \to \{ 0, \dots, \ell C \}$. This method is an
algorithm of stepwise refinement, where at each step a choice takes
place that never needs to be taken back.

\medskip
A system of \emph{free half-flows} is then a sequence
$(\fmin_k, \fmax_k)_{0 \leq k \leq \ell}$ of functions that satisfy
the conditions of Lemma~\ref{thm:half-flows-recursion}
and~\ref{thm:half-flow-bounds}. This means that
\begin{subequations}
  \label{eq:free-halfflows}
  \begin{alignat}{3}
    \label{eq:lower-raising-recursion2}
    \fmin_k(w_{1:k}) &\leq \fmin_{k+1}(w)
    \leq \fmin_k(w_{0:k}) + \# w_k, \\
    \label{eq:upper-raising-recursion2}
    \fmax_k(w_{0:k}) - \#^c w_k &\leq \fmax_{k+1}(w)
    \leq \fmax_k(w_{1:k}) \\
    \intertext{must be true for $0 \leq k < \ell$ and $w \in
      \Sigma^{k+1}$ and}
    \label{eq:fmin-bounds2}
    0 &\leq \fmin_k(v) \leq \#v, \\
    \label{eq:fmax-bounds2}
    0 &\leq \fmax_k(v) \leq \#v + (\ell - k)C, \\
    \label{eq:fmin-fmax-bounds2}
    \fmin_k(v) &\leq \fmax_k(v) \leq \fmin_k(v) + (\ell - k)C
  \end{alignat}
\end{subequations}
must be true for $0 \leq k \leq \ell$ and $v \in \Sigma^k$. (To
distinguish between the two kinds of half-flows, we write free
half-flows with a bar and bound half-flows a tilde.)

\medskip
This definition ensures that each system of bound half-flows is also a
system of free half-flows and also that free half-flows can be chosen
without reference to a given flow function.

The following theorem then states how flow functions are constructed.

\begin{theorem}[Flow construction]
  \label{thm:flow-construction}
  All solutions of the flow
  conditions~\eqref{eq:forward-flow-conditions} can be found by
  constructing a sequence of free half-flows
  \begin{equation}
    \label{eq:half-flow-sequence}
    \fmin_0, \fmax_0, \fmin_1, \fmax_1,
    \dots, \fmin_\ell, \fmax_\ell \,.
  \end{equation}
  This means that if $j < \ell$ and $\fmin_0$, $\fmax_0$, $\fmin_1$,
  $\fmax_1$, $\dots$, $\fmin_j$, $\fmax_j$ is a sequence of
  functions\footnote{The empty sequence is here included.} satisfying
  the conditions in~\eqref{eq:free-halfflows}, it is always possible
  to extend it with two functions $\fmin_{j+1}$ and $\fmax_{j+1}$ so
  that the new sequence again satisfies~\eqref{eq:free-halfflows}.

\medskip
  At the end, we have $\fmin_\ell = \fmax_\ell$ and the function
  $f = \fmin_\ell = \fmax_\ell$ is the constructed flow. All flow
  functions can be found in this way.
\end{theorem}

\begin{proof}
  We have already seen that for every flow there is a sequence of free
  half-flows that satisfies the required inequalities. Therefore all
  possible flows can be reached by the construction of the theorem.

  It remains to show that the construction can always be completed. We
  will now follow it step by step and show that at each step is always
  possible.

\medskip
  The construction starts with the choice of two constants $\fmin_0$
  and $\fmax_0$ that satisfy the conditions
  in~\eqref{eq:free-halfflows}. For $k = 0$, the conditions reduce to
  \begin{equation}
    \label{eq:constr-init}
    \fmin_0 = 0
    \qquad\text{and}\qquad
    0 \leq \fmax_0 \leq \ell C,
  \end{equation}
  so the first step can always be done.

\medskip
  Now assume that $k < \ell$, that the half-flows $\fmin_k$ and
  $\fmax_k$ are already constructed and we want to construct
  $\fmin_{k+1}$ and $\fmax_{k+1}$. We must then find solutions for the
  inequalities~\eqref{eq:free-halfflows}. I will now rewrite them in
  the order in which we will try to find solutions for them.

  The result is the following system of inequalities. For any
  $v \in \Sigma^{k+1}$ it contains all conditions on $\fmin_{k+1}(v)$
  and $\fmax_{k+1}(v)$.
  \begin{subequations}
    \label{eq:constr-step}
    \begin{alignat}{2}
      \label{eq:constr-fmin-longer}
      \fmin_k(v_{1:k}) \leq \fmin_{k+1}(v)
      &\leq \fmin_k(v_{0:k}) + \# v_k, \\
      \label{eq:constr-fmin-bounds}
      0 \leq \fmin_{k+1}(v)
      &\leq \#v, \\[1ex]
      \label{eq:constr-fmax-longer}
      \fmax_k(v_{0:k}) - \#^c v_k \leq \fmax_{k+1}(v)
      &\leq \fmax_k(v_{1:k}) \\
      \label{eq:constr-fmax-bounds}
      0 \leq \fmax_{k+1}(v)
      &\leq \#v + (\ell - k - 1)C \\[1ex]
      \label{eq:constr-fmin-fmax}
      \fmin_{k+1}(v) \leq \fmax_{k+1}(v)
      &\leq \fmin_{k+1}(v) + (\ell - k - 1)C,
    \end{alignat}
  \end{subequations}
  There are $\abs{\Sigma}^{k+1}$ of such systems of inequalities, but
  each pair of $\fmin_{k+1}(v)$ and $\fmax_{k+1}(v)$ occurs in only
  one of them. The values of $\fmin_{k+1}(v)$ and $\fmax_{k+1}(v)$ can
  therefore be chosen independently for each $v$.

\medskip
  First we note that each of the requirements
  in~\eqref{eq:constr-step} can be satisfied individually. The only
  cases in which this is not obvious are~\eqref{eq:constr-fmin-longer}
  and~\eqref{eq:constr-fmax-longer}. But when we remove the central
  term in both, a half-flow condition~\eqref{eq:capacities} for
  $\fmin_k$ or $\fmax_k$ remains, which is true by induction. The
  upper and lower bounds on $\fmin_{k+1}$ and $\fmax_{k+1}$ in each of
  the two requirements are therefore consistent with each other and we
  can find solutions for them.

\medskip
  Next we verify that all requirements in~\eqref{eq:constr-step}
  except the right inequality in~\eqref{eq:constr-fmin-fmax} can be
  satisfied together. We do that by choosing
  \begin{align}
    \label{eq:candidate-solution}
    \fmin_{k+1}(v) &= \fmin_k(v_{1:k}),
    &\fmax_{k+1}(v) &= \min\{ \fmax_k(v_{1:k}), \#v + (\ell - k - 1)C \}
  \end{align}
  as our candidate for a solution. This is the ``lazy'' solution in
  which $\fmin_{k+1}(v)$ is as small and $\fmax_{k+1}(v)$ as large as
  possible. The first two requirements in~\eqref{eq:constr-step} are
  then clearly satisfied. For the next two, we need only to verify the
  left inequality of~\eqref{eq:constr-fmax-longer}. It is equivalent
  to the two inequalities
  \begin{align}
    \fmax_k(v_{0:k}) - \#^cv_k &\leq \fmax_k(v_{1:k}),
    &\fmax_k(v_{0:k}) - \#^cv_k &\leq \#v + (\ell - k - 1)C\,.
  \end{align}
  The left inequality is part of~\eqref{eq:constr-fmax-longer} and
  true by induction. For the right inequality, we add on both sides
  $\#^c v_k$ and get $\fmax_k(v_{0:k}) \leq \#v_{0:k} + (\ell - k)C$.
  This is the version of~\eqref{eq:constr-fmax-bounds} with $v_{0:k}$
  instead of $v$ and again true by induction.

\medskip
  What remains is the right inequality in~\eqref{eq:constr-fmin-fmax}.
  It may be violated by the solution
  candidate~\eqref{eq:candidate-solution}. If that is the case, we
  move the candidates for $\fmin_{k+1}(v)$ and $\fmax_{k+1}(v)$
  stepwise towards each other until we have either found a full
  solution or cannot continue. If we cannot continue, $\fmin_{k+1}(v)$
  cannot be made larger or $\fmax_{k+1}(v)$ smaller. That is, one of
  the left inequalities of~\eqref{eq:constr-fmin-longer}
  and~\eqref{eq:constr-fmin-bounds} and one of the right inequalities
  of~\eqref{eq:constr-fmax-longer} and~\eqref{eq:constr-fmax-bounds}
  must have become equalities, or $\fmin_{k+1}(v) = \fmax_{k+1}(v)$.

  Three cases can be excluded. If $\fmin_{k+1}(v) = \fmax_{k+1}(v)$,
  then all of~\eqref{eq:constr-fmin-fmax} is satisfied and we have a
  solution to the full system of inequalities. If
  $\fmax_{k+1}(v) = 0$, then $\fmin_{k+1}(v) = 0$ too and we have
  again a solution. And $\fmin_{k+1}(v) = \# v$ cannot occur without
  $\fmin_{k+1}(v) = \fmin_k(v_{0:k}) + \# v_k$ because by induction,
  $\fmin_k(v_{0:k}) \leq \# v_{0:k}$ and therefore
  $\fmin_k(v_{0:k}) + \# v_k \leq \# v$.

  So we must have $\fmin_{k+1}(v) = \fmin_k(v_{0:k}) + \# v_k$ and
  $\fmax_{k+1}(v) = \fmax_k(v_{0:k}) - \#^c v_k$. But then
  \begin{align*}
    \fmax_{k+1}(v)
    &= \fmax_k(v_{0:k}) - \#^c v_k \\
    &\leq \fmin_k(v_{0:k}) + (\ell - k)C - C + \# v_k \\
    &= \fmin_{k+1}(v) + (\ell - k - 1)C\,.
  \end{align*}
  So the right inequality of~\eqref{eq:constr-fmin-fmax} can always be
  satisfied.

\medskip
  In other words, the step from $\fmin_k$ and $\fmax_k$ to
  $\fmin_{k+1}$ and $\fmax_{k+1}$ is always possible, and $\fmin_\ell$
  and $\fmax_\ell$ can always be constructed. That they are equal
  follows from~\eqref{eq:fmin-fmax-bounds}.
\end{proof}

\section{Examples and diagrams}

We will now, as an illustration for
Theorem~\ref{thm:half-flows-recursion}, construct all
number-conserving elementary cellular automata.

This would lead to a quite voluminous computation if it were done
directly. So I will at first introduce a diagram notation for all the
flows and half-flows of a number-conserving cellular automaton. With
these \emph{box diagrams}, the computation can be shown in a
reasonable amount of space.

\subsection{A rule with a maximal half-flow}

Before we can begin to work with elementary cellular automata, I will
illustrate the method by applying it to a simpler example. We will now
construct a flow function $f$ for which the half-flow $\fmax_0$ takes
the largest possible value. We will do this for the minimal state set
with capacity $C = 1$, i.\,e.\ the set $\Sigma = \{ 0, 1 \}$, so that
every cell contains at most one particle. The diagrams for the
computation are shown in Figure~\ref{fig:sigma-3}.

\medskip
The construction starts with the box at the top; that at the bottom
represents the resulting flow function. I will first describe how the
diagrams must be read and the how the construction proceeds.
\begin{itemize}
\item Each box stands for a pair of half-flow values, $\fmin_k(w)$ and
  $\fmax_k(w)$ for some $w$. The box for $\fmin_k(w)$ and $\fmax_k(w)$
  is called the \emph{$k$-box} for $w$.

  In the construction, it influences the flow values for all
  neighbourhoods with $w$ at its right end. In the figure, one can
  find these neighbourhoods at the bottom, deep below the $k$-box.

\item The values of $\fmin_k(w)$ and $\fmax_k(w)$ are given by the
  position of the lower and the upper edge of the $k$-box for $w$.

\item The values of $w$ for the innermost $k$-boxes in a diagram
  appear at its bottom, below the box.

\item The dots in the bottom diagram stand for the value of $f$ and at
  the same time for the values of the half-flows $\fmin_\ell$ and
  $\fmax_\ell$, since they all are the same.

  Under each dot stands the neighbourhood which it represents. (The
  neighbourhoods are ordered lexicographically by their
  mirror-images.)

\item The grey lines in the diagrams are restrictions for the upper
  and lower edges of the $(k + 1)$-boxes that will be constructed in
  the next step.
\end{itemize}
The construction itself obeys the following rules, which are a
``graphical version'' of Theorem~\ref{thm:flow-construction}.
\begin{itemize}
\item All $k$-boxes except the outermost one must be drawn directly
  inside a $(k+1)$-box. This ensures that the conditions
  $\fmin_k(v_{1:k}) \leq \fmin_{k+1}(v)$ and
  $\fmax_{k+1}(v) \leq \fmax_k(v_{1:k})$
  of~\eqref{eq:constr-fmin-longer} and~\eqref{eq:constr-fmax-longer}
  are satisfied.

\item The maximal height of a $k$-box is $\ell - k$. This ensures
  that~\eqref{eq:constr-fmin-fmax} is satisfied.

\item The upper edge of a $k$-box is not higher than the greatest
  particle content of all the neighbourhoods that it influences.

  This ensures condition~\eqref{eq:constr-fmax-bounds}, because the
  upper edge of the $k$-box $v$ stands for $\fmax_k(v)$, which is the
  maximum of all neighbourhoods of the form $u v$ with
  $u \in \Sigma^{\ell-k}$: Their highest particle content is therefore
  $(\ell - k)C + \# v$, as required in the condition.

\item When a new $k$-box is constructed, its dimensions are partially
  determined by the pair of grey lines in the diagram above it that
  have the same horizontal extension as itself. The new $k$-box must
  have an upper edge that is not lower than the upper grey line, and a
  lower edge that is not higher than the lower grey line.

  This is to make sure that the conditions
  $\fmin_{k+1}(v) \leq \fmin_k(v_{0:k}) + \# v_k$ and
  $\fmax_k(v_{0:k}) - \#^c v_k \leq \fmax_{k+1}(v)$
  of~\eqref{eq:constr-fmin-longer} and~\eqref{eq:constr-fmax-longer}
  are satisfied.

\item Finally, to keep the promise we just made, we need to draw new
  grey lines at the right positions in the new diagram. The grey lines
  that we need to draw are copies of the arrangement of $k$-boxes that
  we just have constructed -- but with modifications. (a) The first
  one is that we create $C$ copies of the arrangement and squeeze them
  horizontally by a factor $C^{-1}$, so that they again fit into the
  same space. One will then be above the neighbourhoods that have 0 as
  its rightmost state, one above those that end with 1, and so on. (b)
  We now move the arrangements upward by rising amounts: The first
  copy not at all, the next one by 1, and so on. (c) Then we replace
  the boxes in all arrangements by pairs of grey lines: the lower edge
  by a grey line at the same position, but the upper edge by an upper
  grey line \emph{that is $C$ positions below}.

\medskip
  To justify these rules, write the inequalities
  from~\eqref{eq:constr-fmin-longer} and~\eqref{eq:constr-fmax-longer}
  that were mentioned above in the form
  \begin{subequations}
    \label{eq:constr-step-mod}
    \begin{alignat}{2}
      \label{eq:constr-fmax-bounds-mod}
      \fmax_{k+1}(v)
      &\geq \fmax_k(v_{0:k}) + \# v_k - C, \\
      \label{eq:constr-fmin-longer-mod}
      \fmin_{k+1}(v)
      &\leq \fmin_k(v_{0:k}) + \# v_k\,.
    \end{alignat}
  \end{subequations}
  The upper inequality describes the position of the upper grey line,
  $\fmax_{k+1}(v)$, in terms if the upper edge of the $k$-box for
  $v_{0:k}$ and other parameters; the lower inequality is about the
  lower grey line and the lower edge of the new $k$-box. We can then
  see that (a) the terms $\fmax_k(v_{0:k})$ and $\fmin_k(v_{0:k})$,
  which stand for the boundaries of the existing $k$-boxes, occur $C$
  times, once for each value of $v_k$, (b) the addition of $\# v_k$
  means that each line pair is is shifted upward by this amount, and
  (c) the subtraction of $C$ moves each upper grey line down by this
  amount.
\end{itemize}

\begin{figure}[!ht]
\vspace*{2mm}
  \centering
  \scalebox{1.1}{
  \begin{tikzpicture}[x=25pt, y=15pt,
    font=\tiny,
    bounds/.style={thin, lightgray, rounded corners=2pt},
    box/.style={draw, thick, fit={#1}},
    box0/.style={box={#1}, inner sep=8pt, rounded corners=3pt},
    box1/.style={box={#1}, inner sep=6pt, rounded corners=3pt},
    box2/.style={box={#1}, inner sep=4pt, rounded corners=3pt}]
    \def\annotate#1#2#3{%
      \draw[thin, gray] (#1) -- (8.5, #2)
          node[right, black, align=left, font=\footnotesize] {#3}}
    \def\downmargin{-0.95}
    \def\arrowlength{.3}
    \def\yaxis#1{\foreach \y in {0, ..., #1}
                   \node at (-.6, \y) {\y};}
    \def\lbounds(#1,#2;#3,#4)#5{
      \def\r{2pt}
      \path (#1,#2) -- ++(-\r,0) coordinate (a)
                    -- ++(0,-\r) coordinate (b);
      \path (#3,#2) -- ++(\r,0) coordinate (d)
                    -- ++(0,-\r) coordinate (c);
      \draw[bounds] (a) -- (b) -- (c) -- (d);
      \path (#1,#4) -- ++(-\r,0) coordinate (a)
                    -- ++(0,\r) coordinate (b);
      \path (#3,#4) -- ++(\r,0) coordinate (d)
                    -- ++(0,\r) coordinate (c);
      \draw[bounds] (a) -- (b) -- (c) -- (d);
    }
    \def\sbounds(#1;#2,#3)#4{
      \lbounds(#1,#2;#1,#3){#4}
    }
    \def\value(#1,#2)#3{
      \draw[fill] (#1, #2) circle (1pt);
      \node (#3) at (#1, \downmargin) {#3};
      }
    \def\zero{\yaxis3
              \node[box0={(0,0) (7,3)}] {};}
    \def\one{\node[box1={(0,0) (3,2)}] {};
             \node[box1={(4,1) (7,3)}] {};}
    \def\two{\node[box2={(0,0) (1,1)}] {};
             \node[box2={(2,1) (3,2)}] {};
             \node[box2={(4,1) (5,2)}] {};
             \node[box2={(6,2) (7,3)}] {};}
    \matrix[row sep=10pt] {
      \zero
      \lbounds(0,0;3,2){0}
      \lbounds(4,1;7,3){1}
      \path (7,3)
          -- +(8pt,8pt) coordinate (a)
          -- +(4pt,1pt) coordinate (b);
      \annotate{a}{4.1}{$\fmax_0=3$};
      \annotate{b}{3.2}{$\fmax_1(1)\geq 3$};
      \path (7,1) -- +(4pt,-2pt) coordinate (c);
      \path (7,0) -- +(8pt,-8pt) coordinate (d);
      \annotate{c}{0.7}{$\fmin_1(1)\leq 1$};
      \annotate{d}{-0.9}{$\fmin_0=0$};
      \\
      \zero\one
      \lbounds(0,0;1,1){00}
      \lbounds(2,1;3,2){10}
      \lbounds(4,1;5,2){01}
      \lbounds(6,2;7,3){11}
      \node at (1.5, \downmargin) {0};
      \node at (5.5, \downmargin) {1};
      \path (7,1) -- +(2pt,-7pt) coordinate (a);
      \annotate{a}{0.35}{This 1-box\dots};
      \\
      \zero\one\two
      \sbounds(0;0,0){000}
      \sbounds(1;1,1){100}
      \sbounds(2;1,1){010}
      \sbounds(3;2,2){110}
      \sbounds(4;1,1){001}
      \sbounds(5;2,2){101}
      \sbounds(6;2,2){011}
      \sbounds(7;3,3){111}
      \node at (0.5, \downmargin) {00};
      \node at (2.5, \downmargin) {11};
      \node at (4.5, \downmargin) {01};
      \node at (6.5, \downmargin) {11};
      \path (7,2) -- +(0,-4pt) coordinate (a);
      \annotate{a}{1.4}{$\fmin_2(11)=2$};
      \\
      \zero\one\two
      \value(0,0){000}
      \value(1,1){100}
      \value(2,1){010}
      \value(3,2){110}
      \value(4,1){001}
      \value(5,2){101}
      \value(6,2){011}
      \value(7,3){111}
      \path (7,3) -- +(2pt,0pt) coordinate (a);
      \annotate{a}{3.4}{$f_3(111)=3$};
      \node[fit={(001) (111)}, draw=lightgray, inner sep=-1pt] (b) {};
      \annotate{b}{-0.7}{\dots influences\\
        these neigh-\\bourhoods};
      \\
    };
  \end{tikzpicture} }\vspace*{-2mm}
  \caption{Construction of a flow with $\fmax_0 = 3$, with some
    annotations.}
  \label{fig:sigma-3}\vspace*{-4mm}
\end{figure}

This is the way in which the diagrams in Figure~\ref{fig:sigma-3} were
constructed. As an example for the most difficult part, namely the
construction of the thin lines, we will now look at the first diagram
in detail. In it, we see two pairs of grey lines that form the
``shadows'' of two 1-boxes. Both are smaller versions of the outer
0-box with the upper edge lowered by $C = 1$; the right one is also
moved upward by 1.

We can also see that in the first three diagrams of
Figure~\ref{fig:sigma-3}, the arrangement of grey lines in the right
half is a shifted version of the arrangement of grey lines in the
right half.

In the diagrams of Figure~\ref{fig:sigma-3}, the pairs of thin lines
always have the maximal possible distance, so there is no choice in
the following step. This means that the flow function for
$\hat\sigma^3$ is the only one with $\fmax = 3$ -- an observation that
can be easily extended to all $\hat\sigma^\ell$ with $\ell \geq 0$.

\subsection{Elementary cellular automata}

Now we can begin with the construction of all number-conserving
elementary cellular automata. Elementary cellular automata, or ECA
\cite{Wolfram1983}, are simply the one-dimensional cellular automata
with radius $r_1 = r_2 = 1$ and state set $\Sigma = \{ 0, 1 \}$. They
therefore have a transition rule $\phi\colon \Sigma^3 \to \Sigma$ and
a flow function of width $\ell = 2$ for capacity $C = 1$. We must now
find all number-conserving flows.

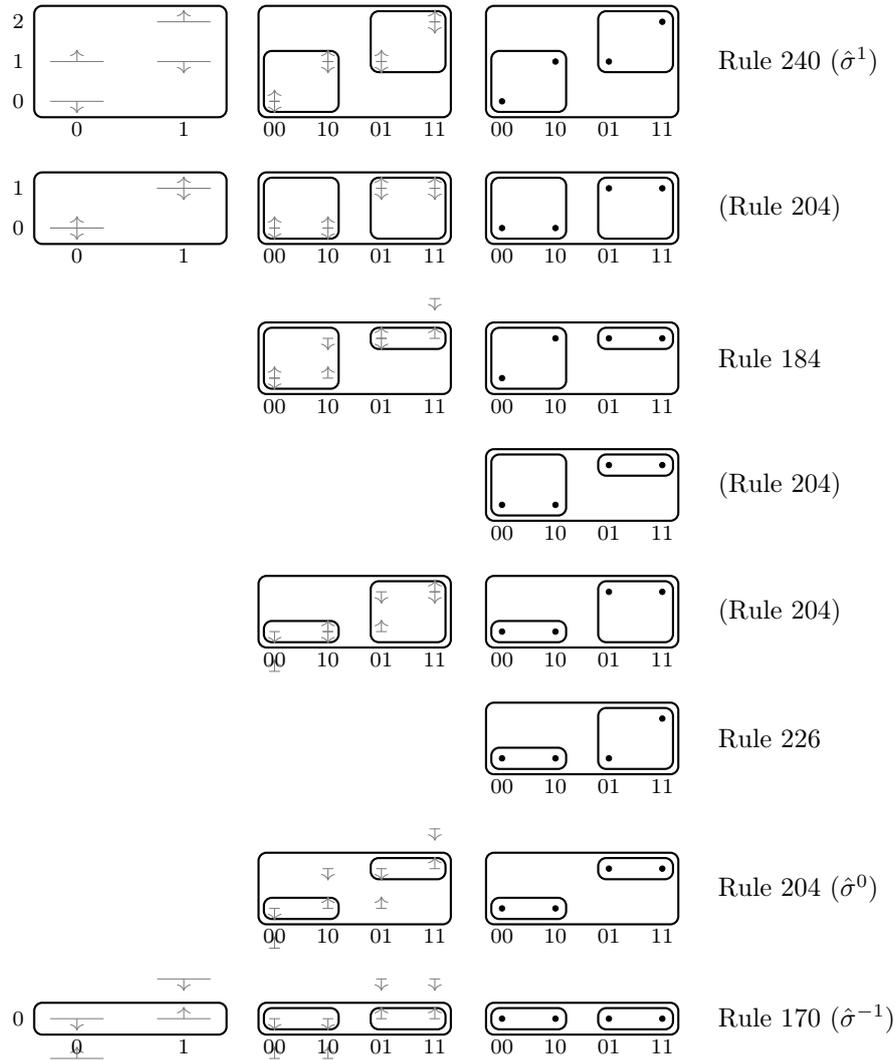
\begin{figure}[t]
  \centering
  \begin{tikzpicture}[x=20pt, y=15pt,
    font=\tiny,
    caption/.style={font=\footnotesize,anchor=west},
    bounds/.style={thin, lightgray, rounded corners=2pt},
    box/.style={draw, thick, fit={#1}},
    box0/.style={box={#1}, inner sep=6pt, rounded corners=3pt},
    box1/.style={box={#1}, inner sep=4pt, rounded corners=3pt}]
    \def\downmargin{-.7}
    \def\arrowlength{.3}
    \def\yaxis#1{\foreach \y in {0, ..., #1}
                   \node at (-.55, \y) {\y};}
    \def\lbounds(#1,#2;#3,#4)#5{
      \def\r{2pt}
      \path (#1,#2) -- ++(-\r,0) coordinate (a)
                    -- ++(0,-\r) coordinate (b);
      \path (#3,#2) -- ++(\r,0) coordinate (d)
                    -- ++(0,-\r) coordinate (c);
      \draw[bounds] (a) -- (b) -- (c) -- (d);
      \path (#1,#4) -- ++(-\r,0) coordinate (a)
                    -- ++(0,\r) coordinate (b);
      \path (#3,#4) -- ++(\r,0) coordinate (d)
                    -- ++(0,\r) coordinate (c);
      \draw[bounds] (a) -- (b) -- (c) -- (d);
    }
    \def\sbounds(#1;#2,#3)#4{
      \lbounds(#1,#2;#1,#3){#4}
    }
    \def\value(#1,#2)#3{
      \draw[fill] (#1, #2) circle (1pt);
      \node at (#1, \downmargin) {#3};
    }
    \def\neighbourhoods{%
      \node at (0.5, \downmargin) {0};
      \node at (2.5, \downmargin) {1};
    }
    \def\zeroa{\node[box0={(0,0) (3,2)}] {};}
    \def\oneaa{\node[box1={(0,0) (1,1)}] {};
               \node[box1={(2,1) (3,2)}] {};}
    \def\zerob{\node[box0={(0,0) (3,1)}] {};}
    \def\oneba{\node[box1={(0,0) (1,1)}] {};
               \node[box1={(2,0) (3,1)}] {};}
    \def\onebb{\node[box1={(0,0) (1,1)}] {};
               \node[box1={(2,1) (3,1)}] {};}
    \def\onebc{\node[box1={(0,0) (1,0)}] {};
               \node[box1={(2,0) (3,1)}] {};}
    \def\onebd{\node[box1={(0,0) (1,0)}] {};
               \node[box1={(2,1) (3,1)}] {};}
    \def\zeroc{\node[box0={(0,0) (3,0)}] {};}
    \def\oneca{\node[box1={(0,0) (1,0)}] {};
               \node[box1={(2,0) (3,0)}] {};}
    \matrix(boxes)[row sep=10pt, column sep=10pt] {
      \yaxis2 \zeroa
      \lbounds(0,0;1,1){0};
      \lbounds(2,1;3,2){1};
      &
      \zeroa\oneaa
      \sbounds(0;0,0){00};
      \sbounds(1;1,1){10};
      \sbounds(2;1,1){01};
      \sbounds(3;2,2){11};
      \neighbourhoods
      &
      \zeroa\oneaa
      \value(0,0){00};
      \value(1,1){10};
      \value(2,1){01};
      \value(3,2){11};
      &
      \node[caption] at (0, 1) {Rule 240 ($\hat\sigma^1$)};
      \\
      \yaxis1 \zerob
      \lbounds(0,0;1,0){0};
      \lbounds(2,1;3,1){1};
      &
      \zerob\oneba
      \sbounds(0;0,0){00};
      \sbounds(1;0,0){10};
      \sbounds(2;1,1){01};
      \sbounds(3;1,1){11};
      \neighbourhoods
      &
      \zerob\oneba
      \value(0,0){00};
      \value(1,0){10};
      \value(2,1){01};
      \value(3,1){11};
      &
      \node[caption] at (0, .5) {(Rule 204)};
      \\
      &
      \zerob\onebb
      \sbounds(0;0,0){00};
      \sbounds(1;1,0){10};
      \sbounds(2;1,1){01};
      \sbounds(3;2,1){11};
      \neighbourhoods
      &
      \zerob\onebb
      \value(0,0){00};
      \value(1,1){10};
      \value(2,1){01};
      \value(3,1){11};
      &
      \node[caption] at (0, .5) {Rule 184};
      \\
      &
      &
      \zerob\onebb
      \value(0,0){00};
      \value(1,0){10};
      \value(2,1){01};
      \value(3,1){11};
      \neighbourhoods
      &
      \node[caption] at (0, .5) {(Rule 204)};
      \\
      &
      \zerob\onebc
      \sbounds(0;0,-1){00};
      \sbounds(1;0,0){10};
      \sbounds(2;1,0){01};
      \sbounds(3;1,1){11};
      \neighbourhoods
      &
      \zerob\onebc
      \value(0,0){00};
      \value(1,0){10};
      \value(2,1){01};
      \value(3,1){11};
      &
      \node[caption] at (0, .5) {(Rule 204)};
      \\
      &
      &
      \zerob\onebc
      \value(0,0){00};
      \value(1,0){10};
      \value(2,0){01};
      \value(3,1){11};
      &
      \node[caption] at (0, .5) {Rule 226};
      \\
      &
      \zerob\onebd
      \sbounds(0;0,-1){00};
      \sbounds(1;1,0){10};
      \sbounds(2;1,0){01};
      \sbounds(3;2,1){11};
      \neighbourhoods
      &
      \zerob\onebd
      \value(0,0){00};
      \value(1,0){10};
      \value(2,1){01};
      \value(3,1){11};
      &
      \node[caption] at (0, .5) {Rule 204 ($\hat\sigma^0$)};
      \\
      \yaxis0 \zeroc
      \lbounds(0,0;1,-1){0};
      \lbounds(2,1;3,0){1};
      &
      \zeroc\oneca
      \sbounds(0;0,-1){00};
      \sbounds(1;0,-1){10};
      \sbounds(2;1,0){01};
      \sbounds(3;1,0){11};
      \neighbourhoods
      &
      \zeroc\oneca
      \value(0,0){00};
      \value(1,0){10};
      \value(2,0){01};
      \value(3,0){11};
      &
      \node[caption] at (0, 0) {Rule 170 ($\hat\sigma^{-1}$)};
      \\
    };
  \end{tikzpicture}
  \caption{Construction of all number-conserving ECA rules.}
  \label{fig:elementary-conserving-automata}
\end{figure}
This is done in Figure~\ref{fig:elementary-conserving-automata}. The
figure must be read from left to right and from top to bottom. The
first column contains all possible box diagrams for the first step of
the construction. To the right of each diagram and right-below of it
are its refinements, the diagrams for all the possible next steps of
the construction. So the second box diagram in the first column has
four refinements, and the second of them, two.

The rightmost column contains the code numbers of the ECA that belong
to the flows in column 3. The code numbers were given by Wolfram
\cite{Wolfram1983}. If a rule is a shift rule $\hat\sigma^k$, the name
of the rule is also given. Since we have here switched back to a
symmetrical cellular neighbourhood, the shift speed is different from
that what we would have seen in the context of
Figure~\ref{fig:sigma-3}.

Some rule names are put in braces. They belong to rules that occur
more than once in the diagram, constructed in different ways. For an
unknown reason it is only Rule 204, the identity function, that occurs
more than once. But there is only one construction of Rule 204 that
uses bound half-flows; it is that one in which the rule name is not
put in braces.

One new phenomenon appears here: It can happen that the ``upper'' grey
line is below the ``lower'' grey line. This occurs the first time in
the construction of Rule 184. It is the reason why the two types of
grey lines have distinct shapes.

The result of this calculation is that the only number-conserving ECA
rules are the shift rules $\hat\sigma^{-1}$, $\hat\sigma^0$,
$\hat\sigma^1$, the so-called ``traffic rule'' 184, and Rule 226,
which is Rule 184 with left and right exchanged.

\section{Minimal flows}
\label{sec:minimal-flows}

We have now found all flow functions with rule width $\ell = 2$ and
for a minimal state set $\Sigma$ of capacity $C = 1$. Similar
computations for larger $\ell$ or $\Sigma$ would soon become unwieldy.
I will therefore now describe another method to construct all flows
for a given $\ell$ and $\Sigma$. It is based on a lattice structure on
the set of flows.

The lattice structure also provides formulas that can serve as names
for all the flow functions.

\subsection{Sets of flows}

To keep the following arguments short (or at least not too long), we
will introduce some notation. We will write $\F(\ell, \Sigma)$ for the
\emph{set of flows} of width $\ell$ and state set $\Sigma$, i.\,e.\
the set of functions
$f \colon \Sigma^\ell \to \{ 0, \dots, \ell C \}$, with
$C = \max \set{ \#\alpha \colon \alpha \in \Sigma}$ that satisfy the
flow conditions~\eqref{eq:forward-flow-conditions}. (The restriction
to $r_1 = \ell$ and $r_2 = 0$ is again in place.)

We will also write $\H(\ell, \Sigma)$ for the \emph{set of free
  half-flows} related to $\ell$ and $\Sigma$. It is the set of
function families $(\fmin_k, \fmax_k)_{0 \leq k \leq \ell}$ that
satisfy the inequalities~\eqref{eq:free-halfflows}.

If $\Sigma$ is a minimal state set of capacity $C$, we will also write
these sets as $\F(\ell, C)$ and $\H(\ell, C)$. When the intended
meaning is clear, we may also write $\F$ or $\H$.

\subsection{Flows as a partially ordered set}

Our next step is the introduction of a partial order on $\F$ and $\H$.
We define the order pointwise: For two flows $f$, $g \in \F$, the
relation $f \leq g$ shall be true if
\begin{equation}
  \label{eq:flow-less}
  \forall v \in \Sigma^\ell \colon f(v) \leq g(v)\,.
\end{equation}
For systems of half-flows $f$, $g \in \H$, conditions similar
to~\eqref{eq:flow-less} must be satisfied for all of their half-flows
to make $f \leq g$ true.

\medskip
The \emph{minimum} $f \meet g$ and \emph{maximum} $f \join g$ of two
flow functions $f$, $g \in \F$ is defined pointwise too, in the form
\begin{subequations}
  \label{eq:meet-join-def}
  \begin{align}
    \label{eq:meet-def}
    \forall v \in \Sigma^\ell \colon
    (f \meet g)(v) &= \min\{f(v), g(v)\}, \\
    \label{eq:join-def}
    \forall v \in \Sigma^\ell \colon
    (f \join g)(v) &= \max\{f(v), g(v)\}\,.
  \end{align}
\end{subequations}
For systems of half-flows, $f \meet g$ and $f \meet h$ are defined in
the same way -- here the formulas of~\eqref{eq:meet-join-def} are
applied to all pairs of half-flows in $f$ and $g$ with the same type.

\medskip
The minimum and maximum of a set $S$ of flows (or half-flow systems)
is written $\Meet S$ and $\Join S$. For $S \subseteq \F$ we have then
\begin{subequations}
  \label{eq:set-meet-join-def}
  \begin{align}
    \label{eq:set-meet-def}
    \forall v \in \Sigma^\ell \colon
    (\Meet S)(v) &= \min\set{ f(v) \colon f \in S }, \\
    \label{eq:set-join-def}
    \forall v \in \Sigma^\ell \colon
    (\Join S)(v) &= \max\set{ f(v) \colon f \in S },
  \end{align}
\end{subequations}
and again the formulas for half-flows are similar.

\medskip
The following theorem shows that these definitions are useful.
\begin{theorem}
  \label{thm:lattice}
  With the operations $\meet$ and $\join$ as defined above, the sets
  $\F(\ell, \Sigma)$ and $\H(\ell, \Sigma)$ each form a distributive
  lattice.
\end{theorem}

\begin{proof}
  We will prove the theorem for $\F$ and then sketch a similar argument
  for $\H$.

  First we note that the set $\{ 0, \dots, C \}^{\Sigma^\ell}$ of
  functions from $\Sigma^\ell$ to $\{ 0, \dots, C \}$ is a
  distributive lattice. This is because $\{ 0, \dots, C \}$, as a
  linear order, is distributive, and
  $\{ 0, \dots, C \}^{\Sigma^\ell}$, as a product of distributive
  lattices, is so too \cite[Proposition 4.8]{Davey2002}.

  It is therefore enough to show that $\F$ is closed under minimum and
  maximum; then it is a sublattice of the full function space and also
  distributive \cite[Section 4.7]{Davey2002}.

\medskip
  To do this, let $f$, $g \in \F$ be two flows and $h = f \meet g$. We
  have then to verify that the following inequalities are true for all
  $v \in \Sigma^\ell$ and $w \in \Sigma^{\ell + 1}$:
  \begin{subequations}
    \label{eq:meet-conditions}
    \begin{align}
      \label{eq:meet-greater}
      h(v) &\geq 0, \\
      \label{eq:meet-less}
      h(v) &\leq \# v, \\
      \label{eq:meet-notempty}
      h(w_{1:\ell}) - h(w_{0:\ell}) &\leq \# w_\ell, \\
      \label{eq:meet-notfull}
      h(w_{0:\ell}) - h(w_{1:\ell}) &\leq \#^c w_\ell\,.
    \end{align}
  \end{subequations}
  The first two inequalities control the upper and lower bounds on $h$
  and are clearly satisfied. The other two represent the flow
  conditions~\eqref{eq:forward-flow-conditions} for $h$. We first look
  at~\eqref{eq:meet-notempty}. The two flow values at its left side
  each stand either for a value of $f$ or of $g$. If both belong to
  the same flow, say $f$, then~\eqref{eq:meet-notempty} is actually
  the same inequality for that flow and therefore true. The only
  interesting case is therefore that in which the two $h$ values stem
  from different flows. Without loss of generality we may assume that
  the first one is from $f$ and the second from $g$, so that we need
  to find an upper bound for $f(w_{1:\ell}) - g(w_{0:\ell})$. But
  since $h(w_{1:\ell}) = f(w_{1:\ell})$, we must have
  $f(w_{1:\ell}) \leq g(w_{1:\ell})$ and can calculate
  \begin{equation}
    \label{eq:meet-nonempty-true}
    f(w_{1:\ell}) - g(w_{0:\ell}) \leq
    g(w_{1:\ell}) - g(w_{0:\ell}) \leq \# w_\ell,
  \end{equation}
  so~\eqref{eq:meet-notempty} must be true.
  Inequality~\eqref{eq:meet-notfull} has the same form and can be
  proved in the same way. This proves that $f \meet g \in \F$.

\medskip
  When we set instead $h = f \join g$, the proof
  of~\eqref{eq:meet-notempty} is a bit different. Now we conclude
  instead from $h(w_{0:\ell}) = g(w_{0:\ell})$ that
  $g(w_{0:\ell}) \geq f(w_{0:\ell})$,
  and~\eqref{eq:meet-nonempty-true} becomes
  \begin{equation}
    \label{eq:join-nonempty-true}
    f(w_{1:\ell}) - g(w_{0:\ell}) \leq
    f(w_{1:\ell}) - f(w_{0:\ell}) \leq \# w_\ell\,.
  \end{equation}
  The rest of the argument is the same, and we have now proved that
  $\F$ is a distributive lattice.

\medskip
  The space $\H$ of free half-flows is a subset not of a single
  function space but the product of several, one for each half-flow.
  This product is still a distributive lattice. Therefore it again is
  enough to check whether it is closed under $\meet$ and $\join$.
  There are a lot more inequalities to consider, but they all can be
  brought to one of the three forms
  \begin{subequations}
    \label{eq:general-form}
    \begin{align}
      \label{eq:general-less}
      h_k(w) &\leq K, \\
      \label{eq:general-greater}
      h_k(w) &\geq K, \\
      \label{eq:general-maxdist}
      h_j(v) - h_k(w) &\leq K\,.
    \end{align}
  \end{subequations}
  In these inequalities, the $h$ terms stand either for $\hmin$ or
  $\hmax$, with possibly different choices in the same
  inequality,\footnote{Different choices are needed because the
    sequence $(\o h_k, \u h_k)_k$ must also satisfy the right
    inequality of~\eqref{eq:fmin-fmax-bounds2}, which here becomes
    $\o h_k(v) \leq \u h_k(v) + (\ell - k) C$. To bring it to the
    form~\eqref{eq:general-maxdist}, it must be written as
    $\o h_k(v) - \u h_k(v) \leq (\ell - k) C$.} $K$ is a constant that
  does not depend on the half-flow functions, and $v \in \Sigma^j$ and
  $w \in \Sigma^k$. These inequalities have the same form as those for
  $\F$, therefore the same arguments can be used, and $\H$ too is a
  distributive lattice.
\end{proof}

\subsection{Minimal flows as building blocks for all the flows}

Now, with the lattice structures of $\F$ and $\H$, we can use a subset
of all flows as building blocks for the rest. These are the
\emph{minimal flows}
\begin{equation}
  \label{eq:minimal-flow}
  m(v, k) = \Meet \set{ f \in \F \colon f(v) \geq k }
\end{equation}
with $v \in \Sigma^\ell$ and $k \in \{0, \dots, \# v \}$. So $m(v, k)$
is the smallest flow that at the neighbourhood $v$ has at least the
strength $k$.

\medskip
Every flow can then be represented as a maximum of minimal flows,
\begin{equation}
  \label{eq:representation}
  f = \Join \set{ m(v, f(v)) \colon v \in \Sigma^\ell }\,.
\end{equation}
To see why this is so, we note that $m(v, f(v))$ is a flow that is
less than or equal than $f$ and agrees with $f$ when evaluated at $v$.
This is true because for $m(v, f(v))$, the right side
of~\eqref{eq:minimal-flow} becomes
$\Meet \set{ g \in \F \colon g(v) \geq f(v) }$. This expression is the
minimum of a set of flows which contains $f$, therefore
$m(v, f(v)) \leq f$. And it is the minimum of a set of functions whose
values at $v$ are greater or equal to $f(v)$ and which contains~$f$,
therefore $m(v, f(v))(v) = f(v)$.

The first fact proves that the right side of~\eqref{eq:representation}
is less than or equal to $f$; the second fact, that for each
neighbourhood $v$, the right side of~\eqref{eq:representation} has a
value that is not less than $f(v)$. So the right side
of~\eqref{eq:representation} must be $f$.

Apart from being building blocks, the minimal flows are also
interesting in their own right. They answer the questions: If I
require that $f(v) = k$, which influence has this on other
neighbourhoods? Does pushing the particles forward in one place set
particles in other place in motion? We will therefore construct the
minimal flows.

\begin{theorem}
  \label{thm:minimal-flows}
  Let $f = m(a, k)$ be a minimal flow in $\F(\ell, \Sigma)$ and
  $b \in \Sigma^\ell$. Then $f(b)$ is the smallest non-negative number
  which satisfies the inequality
  \begin{equation}
    \label{eq:minimal-explicit}
    f(b) \geq k - \min\{ \abs{u}, \abs{u'} \} C - \# w - \#^c w',
  \end{equation}
  for all $u$, $v$, $w$, $u'$, $w' \in \Sigma^*$ with $a = u v w$ and
  $u' v w' = b$.
\end{theorem}

\begin{proof}
  According to Theorem~\ref{thm:flow-construction}, $m(a, k)$ can be
  constructed from free half-flows. We take now the minimum of all
  half-flow sequences that lead to $m(a, k)$, i.\,e.\ the system
  \begin{equation}
    \label{eq:minimal-half-flow}
    M(a, k) = \Meet \set{
      (\fmin_k, \fmax_k)_k \in \H \colon
      \fmin_\ell(a) = \fmax_\ell(a) \geq k }\,.
  \end{equation}
  It consists of a sequence of half-flows that leads to $m(a, k)$
  ``greedily'': $M(a, k)$ is a sequence
  $(\fmin_i, \fmax_i)_{0 \leq i \leq \ell}$ of half-flow pairs in
  which every value $\fmin(v)$ and $\fmax(v)$ is the smallest possible
  for which the construction sequence can still end in $m(a, k)$. We
  will now construct such a sequence.

\medskip
  To do this, we will use the inequalities
  of~\eqref{eq:free-halfflows}, but in a much more compressed form. In
  a first simplification, we write them as
  \begin{subequations}
    \label{eq:short-construction}
    \begin{align}
      \label{eq:short-fmin}
      \fmin(w) &\leq \fmin(v w) \leq \fmin(v) + \#^c w, \\
      \label{eq:short-fmax}
      \fmax(v) - \#^c w &\leq \fmax(v w) \leq \fmax(w), \\
      \label{eq:short-bounds}
      \fmin(v) &\leq \fmax(v) \leq \fmin(v) + (\ell - \abs{v}) C,
    \end{align}
  \end{subequations}
  and they are valid for all $v$, $w \in \Sigma^*$ with
  $\abs{v w} \leq \ell$.

\medskip
  These inequalities differ from their representation
  in~\eqref{eq:free-halfflows} in two aspects: (a) The indices on
  $\fmin$ and $\fmax$ are dropped, since they can be derived from
  their arguments, and writing them would make the formulas only more
  complicated. (b) In the first two lines, the formulas have been
  \emph{iterated} and the variables \emph{renamed}. From the
  inequality $\fmin_{k+1}(w) \leq \fmin_k(w_{0:k}) + \# w_k$
  in~\eqref{eq:lower-raising-recursion}, valid for
  $w \in \Sigma^{k+1}$, we get by induction
  $\fmin_{k+n}(w) \leq \fmin_k(w_{0:k}) + \# w_{k:n}$ for
  $w \in \Sigma^{k+n}$, and then, after renaming $w_{0:k}$ and
  $w_{k:n}$ to $v$ and $w$, the right side of~\eqref{eq:short-fmin}.
  The other derivations are similar.

  You may also have noticed that conditions~\eqref{eq:fmin-bounds2}
  and \eqref{eq:fmax-bounds2} have disappeared. They are less
  important. Their left parts state that all half-flows must be
  non-negative, which we will keep in mind, while their right parts
  contain upper bounds to the half-flows. We will not need them
  because in our construction, all half-flows are as small as
  possible.

\medskip
  In a second compression step, we now express the
  inequalities~\eqref{eq:short-construction} with arrows. We write an
  inequality
  \begin{equation}
    \label{eq:bounds-inequality}
    \fmin(u) + n \leq \fmax(v),
  \end{equation}
  as an arrow
  \begin{equation}
    \label{eq:bounds-arrow}
    \u u \arrow{n} \o v,
  \end{equation}
  and do the same for all other combinations of upper and lower bars.
  The number $n$ on an arrow is called its \emph{strength}, and it is
  omitted when $n = 0$.

\medskip
  An arrow as in~\eqref{eq:bounds-arrow} can be interpreted as ``if
  $\fmin(u) \geq x$, then $\fmax(v) \geq x + n$''. So when we have a
  chain of arrows, like $\o u \arrow{m} \o v \arrow{n} \u w$ we can
  add their strengths and get a new arrow, in this case
  $\o u \arrow{m + n} \u w$. We want to find arrows of the form
  $\u a \arrow{n} \u b$, because they lead to lower bounds on $f(b)$.
  (Note that, since $\abs{a} = \abs{b} = \ell$, we have
  $\fmin(a) = \fmax(a)$ and $\fmin(b) = \fmax(b)$, so that upper and
  lower bars on $a$ and $b$ are interchangeable.)

\medskip
  When we now translate the inequalities
  of~\eqref{eq:short-construction} into arrows, they become
  \begin{subequations}
    \label{eq:arrows}
    \begin{align}
      \label{eq:lower-arrows}
      \u w &\arrow{} \u{v w}, &
      \u{v w} &\arrow{-\# w} \u v, \\
      \label{eq:upper-arrows}
      \o v & \arrow{-\#^c w} \o{v w}, &
      \o{v w} &\arrow{} \o w, \\
      \label{eq:updown-arrows}
      \u v & \arrow{} \o v, &
      \o v & \arrow{(\abs{v} - \ell) C} \u v \,.
    \end{align}
  \end{subequations}

  In principle, we must consider all arrow chains from $a$ to $b$. But
  most of them can be replaced by stronger chains, and only two
  remain. The following arguments show how this is done.
  \begin{enumerate}
  \item \emph{We can ignore all chains in which two arrows of the same
      form occur in sequence.}

    An example for such a chain is
    $\u{v w w'} \arrow{-\# w'} \u{v w} \arrow{-\# w} \u v$, in which
    the left arrow of~\eqref{eq:lower-arrows} occurs twice. It can be
    replaced with the equally strong arrow
    $\u{v w w'} \arrow{-\# w w'} \u v$. The same can be done with the
    other arrows in~\eqref{eq:lower-arrows}
    and~\eqref{eq:upper-arrows}, which are the only arrows to which
    this rule applies.

  \item In the next reduction step, we consider arrow chains that
    consists of arrows of the same ``type''. Two arrows have the same
    type if they either lead from an upper half-flow to an upper
    half-flow or from a lower to an lower half-flow. \emph{Any chain
      of same-type arrows can be brought to one of the forms
      \begin{subequations}
        \label{eq:same-type}
        \begin{gather}
          \label{eq:same-type-lower}
          \u{v w} \arrow{-\#w} \u v \arrow{} \u{u v}, \\
          \label{eq:same-type-higher}
          \o{u v} \arrow{} \o v \arrow{-\#^c w} \o{vw}
        \end{gather}
      \end{subequations}
      without loss of strength.}
\eject

    In other words, we can assume that a shortening arrow always
    occurs before a lengthening one.

    To prove this, we first show that if a lengthening arrow is
    followed by a shortening arrow, they can be rearranged without
    loss of strength. For lower half-flows, such a chain must have the
    form (a) $\u u \arrow{} \u{u x v} \arrow{-\# x v} \u v$, or (b)
    $\u{u v} \arrow{} \u{u v w} \arrow{-\# w} \u{v w}$. Form (a)
    occurs when the first arrow adds cell states that the second takes
    away, while (b) occurs when there is a common substring $v$ in the
    first and third half-flow of the chain.

    But in (a), we can remove $x$ and get a stronger arrow chain. This
    chain is the special case of (b) with $v = \epsilon$, so that we
    only need to consider (b). And (b) can be replaced
    with~\eqref{eq:same-type-lower}, so that we have proved our
    assertion for chains of two arrows.

    Chains of more than two arrows can now be rearranged so that there
    is first a chain of shortening and then one of lengthening arrows.
    But arrows of the same form can be condensed to a single arrow, as
    we have seen before. So we end up again
    with~\eqref{eq:same-type-lower}.

    Chains of less than two arrows can be extended by adding ``empty''
    arrows, say by setting $u = \epsilon$. Therefore all sequences of
    arrows between lower half-flows can be brought into the
    form~\eqref{eq:same-type-lower}. The proof for upper half-flows is
    similar.

  \item Another simplification concerns the \emph{type-changing}
    arrows in~\eqref{eq:updown-arrows}. We can assume that \emph{two
      type-changing arrows never occur directly in sequence}. For if
    they occur, as in
    $\u v \arrow{} \o v \arrow{(\abs{v} -\ell) C} \u v$, we can remove
    them both and get an arrow chain that is at least as strong.

  \item Next we consider a chain of three arrows with a type-changing
    arrow in its centre. As we now can conclude, they can only have
    the following two forms (with $u v = u' v'$ or $v u = v' u'$,
    respectively):
    \begin{subequations}
      \label{eq:three-chain}
      \begin{gather}
        \label{eq:three-chain-up}
        \u v \arrow{} \u{u v} \arrow{}
        \o{u' v'} \arrow{} \o{v'}, \\
        \label{eq:three-chain-down}
        \o v \arrow{-\#^c u} \o{v u} \arrow{(\abs{v u} - \ell)C}
        \u{v' u'} \arrow{-\# u'} \u{v'}\,.
      \end{gather}
    \end{subequations}
    But for them we can assume that either $u = \epsilon$ or
    $u' = \epsilon$ and simplify the arrow chains accordingly. (In
    other words, \emph{do not take away what you just have added}.)

\medskip
    The proof consists of four cases. For for each arrow chain, one
    must distinguish between $\abs{u} \geq \abs{u'}$ and
    $\abs{u} \leq \abs{u'}$. We will look only at one case, namely
    that in which $\abs{u} \leq \abs{u'}$ is true
    in~\eqref{eq:three-chain-down}. There we can write $u' = x u$ for
    a suitable $x \in \Sigma^*$,
    \begin{equation}
      \label{eq:three-chain-down1}
      \o v \arrow{-\#^c u} \o{v u} \arrow{(\abs{v u} - \ell)C}
      \u{v' x u} \arrow{-\# x u} \u{v'},
    \end{equation}
    and then remove $u$ to get a stronger chain,
    \begin{equation}
      \label{eq:three-chain-down2}
      \o v \arrow{} \o{v} \arrow{(\abs{v} - \ell)C}
      \u{v' x} \arrow{-\# x} \u{v'}\,.
    \end{equation}
    The other cases are similar and always lead to a new chain that is
    at least as strong as the original one.

\medskip
  \item Now we can construct the two chains that lead to
    inequality~\eqref{eq:minimal-explicit} of the theorem:
    \begin{subequations}
      \label{eq:arrow-bounds}
      \begin{gather}
        \label{eq:arrow-bound-lower}
        \u a = \u{u v w}
        \arrow{-\# w} \u{u v}
        \arrow{} \o{u v}
        \arrow{} \o v
        \arrow{-\#^c w'} \o{v w'}
        \arrow{(\abs{v w'} - \ell) C} \u{v w'}
        \arrow{} \u{u' v w'} = \u b, \\
        \label{eq:arrow-bound-upper}
        \o a = \o{u v w}
        \arrow{} \o{v w}
        \arrow{(\abs{v w} - \ell) C} \u{v w}
        \arrow{-\# w} \u v
        \arrow{} \u{u' v}
        \arrow{} \o{u' v}
        \arrow{-\#^c w'} \o{u' v w'} = \o b\,.
      \end{gather}
    \end{subequations}
    We must collect the arrow chains for all possible $u$, $v$, $w$,
    $u'$, $v' \in \Sigma^*$ to get all requirements on $f(b)$ for a
    given $f(a)$.

    The chains arise naturally once we note that $a$ has maximal
    length and that therefore the first arrow must necessarily be
    shortening. After that, there is only one possible successor for
    each arrow, until the arrow chain ends in $b$.

    The first chain has the weight $-\# w - \#^c w' - \abs{u'} C$ and
    the second, $-\# w - \#^c w' - \abs{u} C$, from which we can see
    that~\eqref{eq:minimal-explicit} is the right formula.

  \item Our argument is not yet complete. The arrow chains
    of~\eqref{eq:arrow-bounds} consist of a single cycle of shortening
    and lengthening arrows, from $u v w$ to $u' v' w'$. What if we
    created arrow chains of more than one such cycle? Would we then
    get more conditions on $f(b)$?

\medskip
    To resolve this question, we first need shorter expressions for
    the cycles. We write the chains in~\eqref{eq:arrow-bounds} as
    single arrows, as
    \begin{subequations}
      \label{eq:cycles}
      \begin{gather}
        \label{eq:cycle-lower}
        \u{u v w}
        \arrow{-\abs{u'} C - \#w - \#^c w'} \u{u' v' w'}, \\
        \label{eq:cycle-upper}
        \o{u v w}
        \arrow{-\abs{u} C - \#w - \#^c w'} \o{u' v' w'}\,.
      \end{gather}
    \end{subequations}
    The two cycles have essentially the same form, so it will be
    enough to consider only the first one.

\medskip
    When we now connect two arrows of the form~\eqref{eq:cycle-lower},
    the result can always be written as
    \begin{equation}
      \label{eq:two-cycles}
      \u{u x y z w}
      \arrow{-\abs{u'} C - \#w - \#^c w'}
      \u{u' x y z w'}
      \arrow{-\abs{u''} C - \#z w - \#^c w''}
      \u{u'' y w''}\,.
    \end{equation}
    This is because the action of each arrow can be understood as
    taking away cell states from both sides of the string at its left
    and then adding others, resulting in the string at its right.
    (In~\eqref{eq:cycles}, the regions $u$ and $w$ are removed and
    $u'$ and $w'$ then added.) A region in the centre is left
    unchanged. With two arrows, the unchanged region of the first
    arrow might be shortened by the second. In~\eqref{eq:two-cycles}
    we therefore have assumed, that $x y z$ is the unchanged region of
    the first and $y$ that of the second arrow.

\medskip
    The strength of the arrow chain in~\eqref{eq:two-cycles} is
    $-\abs{u' w' u''} C - \#z w - \#^c w''$. But we can achieve the
    same result with a single arrow,
    \begin{equation}
      \label{eq:two-cycles-single}
      \u{u x y z w}
      \arrow{-\abs{u''} C - \#z w - \#^c w''}
      \u{u'' y w''}\,.
    \end{equation}
    Its strength differs from that of~\eqref{eq:two-cycles} by
    $-\abs{u' w'}C$, which is never a positive number. This means that
    we can replace~\eqref{eq:two-cycles}
    with~\eqref{eq:two-cycles-single} in a chain of arrows and, by
    induction, that a single cycle~\eqref{eq:cycles} is enough.

\medskip
  \item Finally, what about shorter chains? One could omit the
    type-changing arrows and get chains of the form
    \begin{subequations}
      \label{eq:shortchain}
      \begin{gather}
        \label{eq:shortchain-lower}
        \u a = \u{v w} \arrow{-\# w}
        \u v \arrow{} \u{u' v} = \u b, \\
        \label{eq:shortchain-upper}
        \o a = \o{u v} \arrow{}
        \o v \arrow{-\#^c w'} \o{v w'} = \o b,
      \end{gather}
    \end{subequations}
    which look as if they could be stronger that those
    in~\eqref{eq:arrow-bounds}. But in fact they are just special
    cases of these chains. One can e.\,g.\ see
    that~\eqref{eq:shortchain-lower} is
    just~\eqref{eq:arrow-bound-upper} with $u = w' = \epsilon$. Recall
    that $\u a = \o a$ and $\u b = \o b$ and note that, since
    $\abs{v w} = \ell$, the arrow
    $\o{v w} \arrow{(\abs{v w} - \ell) C} \u{vw}$
    in~\eqref{eq:arrow-bound-upper} has strength 0. In the same way
    one can see that~\eqref{eq:shortchain-upper} is a special case
    of~\eqref{eq:arrow-bound-lower}.
  \end{enumerate}
  We have now shown that all relations between $f(b)$ and $f(a)$
  derive from the arrow chains in~\eqref{eq:arrow-bounds}. Therefore
  the conditions in~\eqref{eq:minimal-explicit} define
  $m(a, k)$.\footnote{This proof is a bit long. Another way to prove
    this theorem -- possibly shorter but less natural -- is to derive
    the inequalities~\eqref{eq:minimal-explicit} only as necessary
    conditions and then to verify that the function $f$ defined by
    them satisfies $f(a) \geq k$ and the flow conditions. If the
    conditions~\eqref{eq:minimal-explicit} were not sufficient, $f$
    would be smaller than all the flows $g$ with $g(a) \geq k$ and
    therefore not be a flow.}
\end{proof}

\subsection{Examples}

With Theorem~\ref{thm:minimal-flows}, we can now find examples for
minimal number-conserving automata. As before with the elementary
cellular automata, we use the minimal state set $\Sigma = \{ 0, 1 \}$.
So we have $C = 1$, and every cell may contain at most one particle.
Even with these restrictions, we can find cellular automata with an
interesting behaviour.

\paragraph{The influence of the particle density}

\begin{figure}[!ht]
\vspace*{3mm}
  \center
 \includegraphics{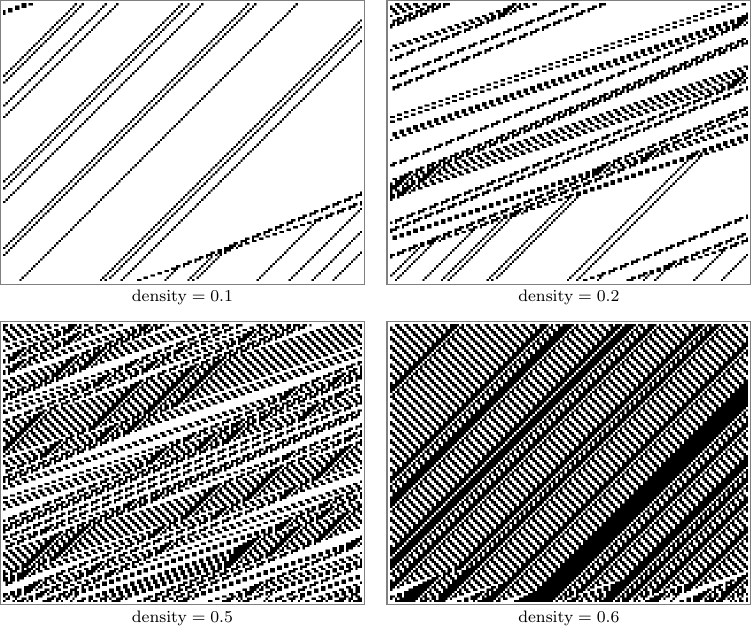}
  \caption{Rule $m(0110, 2)$ at different densities. Time goes
    upward.}
  \label{fig:m0110-2-block}
\end{figure}

An example is $m(0110, 2)$ in Figure~\ref{fig:m0110-2-block}. In the
figure, cells in state 0 are displayed in white, and cells in state 1
in black. Time runs upward, and the line at the bottom is a random
initial configuration.

In an ordinary cellular automaton, one can expect that the number of
ones and zeroes in a random initial configuration has no great
influence on the patterns that arise after a few generation. Under
most rules, the fraction of cells in state 1 changes over time. With
number-conserving automata this is no longer true.
Figure~\ref{fig:m0110-2-block} therefore contains four evolutions of
$m(0110, 2)$ with different densities: the density is the fraction of
cells in state 1 in the initial configuration -- which then stays
constant during the evolution of the cellular automaton.

\medskip
Especially in the two low density evolutions (with density${} = 0.1$
and 0.2), one can see that particles move with three possible speeds:
Isolated particles move with speed 1, blocks of two particles, like
0110, move with speed 3, and blocks of three particles move 5 steps
over two generations (from 0111000000 via 0001101000 to 0000001110)
and therefore have a speed of 2.5. For very low densities like 0.1,
the evolution is initially dominated by non-interacting single
particles, but we can already see how they are collected by the faster
structures and integrated into their particle stream. In the density
0.2 image we can see how the fast structures interact: When a middle
speed particle group interacts with a fast one, a short ``traffic
jam'' of high particle density occurs, but then the two particle
groups separate again. (Note that the middle-speed groups before and
after the interaction consist of different particles.) With density
0.5, the traffic jams are more common, and also highly regular
structures between them. With the moderately high density of 0.6, all
interesting behaviour stops early, and the automaton looks like
$\sigma^1$.

So $m(0110, 2)$ already establishes a vaguely traffic-like behaviour,
except that the speed of a particle also depends on the location of
the particles next to it.

\paragraph{Other phenomena}

\begin{figure}[ht]
  \center
 \includegraphics{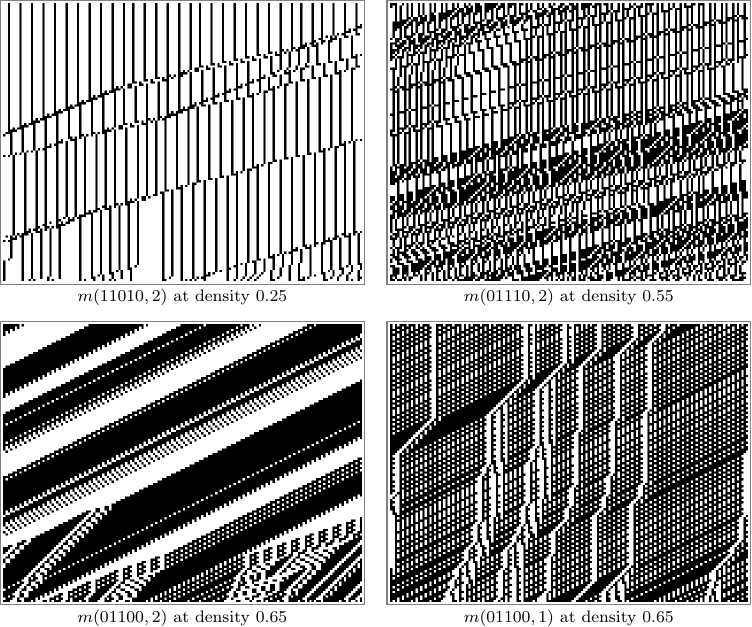}
  \caption{Some minimal rules.}
  \label{fig:examples}\vspace*{-2mm}
\end{figure}

Figure~\ref{fig:examples} illustrates other phenomena that occur with
minimal number-conserving cellular automata.

\medskip
Often, especially at low densities, nothing or almost nothing happens.
The evolution of rule $m(11010, 2)$ at the top left of the figure is
an example. Most of the time, almost all particles are arranged in a
pattern in that no non-zero flow arises. From time to time, a
disturbance moves over the cells.

Rule $m(01110, 2)$, at the top right, consists too of times of
inactivity and moving high-density regions. The high-density regions
have however their own intricate structure.

The evolution of rule $m(01100, 2)$, at the bottom left, has large
regions of the highest possible density, and the disturbances in them
sometimes look like ``anti-particles'': periodic structures of
emptiness between particles.

At the bottom right, the related rule $m(01100, 1)$ shows a
different pattern of almost regular high-density regions. The
disturbances in them show a complex, tree-like pattern.

\section{Odds and ends}

\subsection{Non-deterministic number conservation}

So far, the possibility that two or more states have the same particle
content has rarely been addressed. Such non-minimal state sets have
however an important application: They enable the construction of
non-deterministic number-conserving cellular automata.

\medskip
In a non-deterministic cellular automaton, the value of the transition
function $\phi$ is a set of states, and the next state of a cell may
be any element of the value of $\phi$, when applied to its
neighbourhood. Instead of~\eqref{eq:local-evolution}, we have
\begin{equation}
  \label{eq:local-nondet-evolution}
  \hat{\phi}(a)_x \in \phi(a_{x - r_1}, \dots a_{x + r_2})
  \quad\text{for all $x \in \Z$.}
\end{equation}
When we then have constructed a flow function $f$ for a non-minimal
state set, we can construct a non-deterministic number-conserving rule
with the help of a relaxed form of~\eqref{eq:particles-next-step},
namely
\begin{equation}
  \label{eq:particles-next-step-nondet}
  \phi(w) \subseteq
  \set{ \alpha \in \Sigma \colon
    \# \alpha = f(w_{0:\ell}) + \# w_{r_1} - f(w_{1:\ell})}\,.
\end{equation}

Such a function is clearly number-conserving, since all possible
choices for the next state of a cell have the same particle content.

\medskip
The condition that for each $w \in \Sigma^{\ell + 1}$, all elements of
$\phi(w)$ must have the same particle content is also necessary: If
there are $\alpha$, $\beta \in \phi(w)$ with
$\# \alpha \neq \# \beta$, we can use a configuration $a$ which
contains $w$ as a substring to construct a counterexample. Among the
possible successor configurations of $a$ in the next time step, there
must be two that only differ at one cell, which is in one
configuration in state $\alpha$ and in the other one in state $\beta$.
Since $\# \alpha \neq \# \beta$, number conservation cannot hold for
both configurations.

\subsection{Two-sided neighbourhoods}
\label{sec:two-sided}

We now return to the case where $r_1$ and $r_2$ are arbitrary
non-negative numbers. The flow functions for such \emph{two-sided
  neighbourhoods} are related in a very simple way to the one-sided
neighbourhoods that we have so far investigated.

\medskip
Let $\hat\phi$ be a global transition rule with arbitrary $r_1$ and
$r_2$ and let $\hat\phi'$ be the transition rule with a one-sided
neighbourhood that is related to it. Let $f$ and $f'$ be their flow
functions. Then we can write
\begin{equation}
  \label{eq:neighbourhood-relation}
  \hat \phi = \hat \phi' \circ \hat \sigma^{-r_2}\,.
\end{equation}
This is because we can get the effect of $\hat\phi$ by first moving
the content of all cells by $r_2$ positions to the left and then
applying $\hat \phi'$. The corresponding equation for flows is
\begin{equation}
  \label{eq:flow-relation}
  f(u v) = f'(u v) - \# v,
\end{equation}
for all $u \in \Sigma^{r_1}$ and $v \in \Sigma^{r_2}$: The shift $\hat
\sigma^{-r_2}$ produces an additional flow of $- \# v$ particles over
the boundary between $u$ and $v$.

\medskip
As a corollary of~\eqref{eq:flow-relation}, we see that the new
two-sided flows obey the same partial order as the one-sided do. If
$g$ is another flow with radii $r_1$ and $r_2$ and $f'$ is its
one-sided equivalent, then
\begin{equation}
  \label{eq:order-equivalence}
  f \leq g
  \qquad\text{iff}\qquad
  f' \leq g',
\end{equation}
as we easily can conclude from~\eqref{eq:flow-relation} and the
definition of the partial order in~\eqref{eq:flow-less}.

\medskip
The theory of minimal flows for two-sided neighbourhoods is therefore
isomorphic to that for one-sided neighbourhoods. As a good notation
for two-sided minimal flows I would propose
\begin{equation}
  \label{eq:two-sided-minimal}
  m(u, v; k) = m(u v, k) - \# v\,.
\end{equation}

\section{Open questions}

The set of all number-conserving one-dimensional cellular automata
has, as we have seen, an intricate structure. It leads to many open
questions, of which I will list a few, with comments:
\begin{enumerate}
\item How many number-conserving automata are there for a given state
  set $\Sigma$ and flow length $\ell$?

\item How many flow functions are there for given $\Sigma$ and $\ell$?

  The answers to this and the previous question are the same if
  $\Sigma$ is a minimal state set. For both questions, an explicit
  formula as answer is probably very complex. An asymptotic formula
  could be easier to find and might provide more insight.

  Boccara and Fuk\'s \cite{Boccara1998} have already found that for
  $\Sigma = \{0, 1\}$, there are 5 rules for $\ell = 2$, 22 rules for
  $\ell = 3$ and 428 for $\ell = 4$.

\item If two number-conserving automata have the same flow function,
  how is their behaviour related?

  An answer to this question would provide insight into the relation
  between cellular automata and their flow functions. It would also be
  helpful for the better understanding of non-deterministic
  number-conserving automata.

\item If the behaviour of the automata with flow functions $f$ and $g$
  is known, what can be said about those with flows $f \meet g$ and $f
  \join g$?

  Ideally, the lattice structure of the flows would provide
  information about the cellular automata. This problem should first
  be investigated for minimal state sets, since otherwise it would
  require an answer to the previous question.

\item Given $\fmin_k$ and $\fmax_k$, what can be said about
  $\hat \phi$?

  The pairs $(\fmin_k, \fmax_k)$ provide a classification of
  number-conserving cellular automata, and we should expect that they
  group automata with similar behaviour. But what does ``similar''
  mean in this context?

  An argument similar to that for Figure~\ref{fig:sigma-3} gives a
  partial result: $\fmax_0 = \ell C$ enforces that $\hat \phi$ is the
  shift function $\hat \sigma^\ell$.

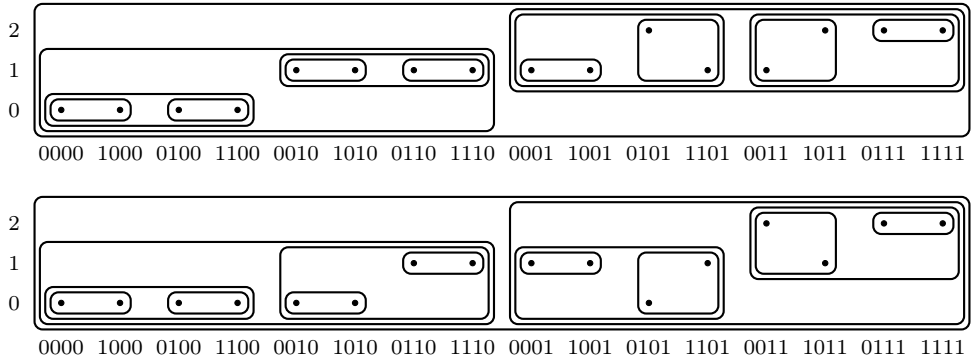
\begin{figure}[t]
  \centering
  \begin{tikzpicture}[x=21.5pt, y=15pt,
    font=\tiny,
    bounds/.style={thin,gray},
    box/.style={draw, thick, fit={#1}},
    box0/.style={box={#1}, inner sep=10pt, rounded corners=3pt},
    box1/.style={box={#1}, inner sep=8pt, rounded corners=3pt},
    box2/.style={box={#1}, inner sep=6pt, rounded corners=3pt},
    box3/.style={box={#1}, inner sep=4pt, rounded corners=3pt}]
    \def\downmargin{-1}
    \def\arrowlength{.3}
    \def\yaxis#1{\foreach \y in {0, ..., #1}
                   \node at (-.65, \y) {\y};}
    \def\value(#1,#2)#3{
      \draw[fill] (#1, #2) circle (1pt);
      \node at (#1, \downmargin) {#3};
      }
    \def\nbox#1(#2,#3;#4,#5){\node[box#1={(#2,#3) (#4,#5)}] {};}
    \matrix[row sep=10pt] {
      \yaxis2
      \value(0,0){0000}
      \value(1,0){1000}
      \nbox3(0,0;1,0);
      \value(2,0){0100}
      \value(3,0){1100}
      \nbox3(2,0;3,0);
      \nbox2(0,0;3,0);
      \value(4,1){0010}
      \value(5,1){1010}
      \nbox3(4,1;5,1);
      \value(6,1){0110}
      \value(7,1){1110}
      \nbox3(6,1;7,1);
      \nbox2(4,1;7,1);
      \nbox1(0,0;7,1);
      \value(8,1){0001}
      \value(9,1){1001}
      \nbox3(8,1;9,1);
      \value(10,2){0101}
      \value(11,1){1101}
      \nbox3(10,1;11,2);
      \nbox2(8,1;11,2);
      \value(12,1){0011}
      \value(13,2){1011}
      \nbox3(12,1;13,2);
      \value(14,2){0111}
      \value(15,2){1111}
      \nbox3(14,2;15,2);
      \nbox2(12,1;15,2);
      \nbox1(8,1;15,2);
      \nbox0(0,0;15,2);
      \\
      \yaxis2
      \value(0,0){0000}
      \value(1,0){1000}
      \nbox3(0,0;1,0);
      \value(2,0){0100}
      \value(3,0){1100}
      \nbox3(2,0;3,0);
      \nbox2(0,0;3,0);
      \value(4,0){0010}
      \value(5,0){1010}
      \nbox3(4,0;5,0);
      \value(6,1){0110}
      \value(7,1){1110}
      \nbox3(6,1;7,1);
      \nbox2(4,0;7,1);
      \nbox1(0,0;7,1);
      \value(8,1){0001}
      \value(9,1){1001}
      \nbox3(8,1;9,1);
      \value(10,0){0101}
      \value(11,1){1101}
      \nbox3(10,0;11,1);
      \nbox2(8,0;11,1);
      \value(12,2){0011}
      \value(13,1){1011}
      \nbox3(12,1;13,2);
      \value(14,2){0111}
      \value(15,2){1111}
      \nbox3(14,2;15,2);
      \nbox2(12,1;15,2);
      \nbox1(8,0;15,2);
      \nbox0(0,0;15,2);
      \\
    };
  \end{tikzpicture}\vspace*{-1mm}
  \caption{Two incomparable flows, $m(0101, 2)$ and $m(0011, 2)$.}
  \label{fig:incomparable}
\end{figure}

\begin{figure}[ht]
\vspace*{2mm}
  \centering
  \includegraphics{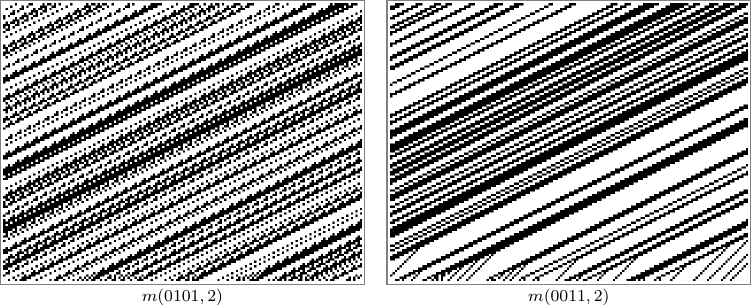}\vspace*{-1mm}
  \caption{The flows $m(0101,2)$ and $m(0011, 2)$ at density 0.4.}
  \label{fig:incomparable-evolution}
\end{figure}
\item What kind of lattices are the flow sets $\F(\ell, \Sigma)$?

  We can see in Figure~\ref{fig:elementary-conserving-automata} that
  $\F(2, \{0, 1\})$ is a linear order. But there are larger values of
  $\ell$ for which $\F(\ell, \{0, 1\})$ does contain incomparable
  elements. One example is shown in Figure~\ref{fig:incomparable}
  and~\ref{fig:incomparable-evolution}.

\item What does the theory for more than one kind of particle look
  like?

  There are two types of multi-particle automata. In automata of the
  first type, each cell contains several containers, each for one kind
  of particle, and the particles only move between ``their''
  containers. This kind of automaton has a theory that is a
  straightforward extension of the one-particle theory.

  In automata of the second type, there is only one container in each
  cell and particles of different ``colours'' that move between the
  the cells. Here the theory will be more complex.

\item What about particles in higher dimensions?

  The derivation of Theorem~\ref{thm:flow-construction} only works in
  one dimension. For higher-dimensional cellular automata therefore
  new ideas are needed.

\item Are there practical or theoretical applications for this theory?

  With practical applications I mean e.\,g.\ simulations of physical
  systems. A theoretical application could be the construction of a
  universal number-conserving cellular automaton or something similar.

\item Can the theory be simplified?

  In the proofs and calculations of this paper, a small number of
  types of inequalities are used over and over again. Is there a
  theory with which the repetitions can be compressed into a few
  lemmas at the beginning, such that the actual proofs take only (say)
  two pages? This would also be helpful for multi-particle and
  higher-dimensional systems.
\end{enumerate}

\paragraph{Acknowledgements} Barbara Wolnik provided helpful comments
on an earlier version of this article. A very thorough anonymous
reviewer also contributed much to the quality of the text.


\end{document}